\documentclass[prl,aps,reprint,a4paper,superscriptaddress,floatfix,longbibliography]{revtex4-2}
\usepackage{times}
\usepackage{graphicx}
\usepackage{epsfig}
\usepackage{amsfonts}
\usepackage{amsmath}
\usepackage{amssymb}
\usepackage{dsfont}
\usepackage{amsthm}
\usepackage{color}
\usepackage{comment}
\usepackage{braket}
\usepackage{physics}
\usepackage{multirow}
\usepackage{diagbox}
\usepackage{multirow}
\usepackage{makecell}
\usepackage{enumerate}

\newcommand{\id}{\mathds{1}}

\usepackage{float}
\usepackage{rotating}
\usepackage[short,c3]{optidef}
\usepackage[linesnumbered,ruled]{algorithm2e}
\usepackage{scalerel}[2016/12/29]

\usepackage[colorlinks=true,linkcolor=blue,citecolor=blue,urlcolor=blue]{hyperref}
\usepackage[justification=raggedright,singlelinecheck=false]{caption}
\usepackage{subcaption}
\newtheoremstyle{italicheader} 
  {0pt} 
  {0pt} 
  {} 
  {\parindent} 
  {\itshape} 
  {.} 
  { } 
  {\itshape #1~#2} 

\renewenvironment{proof} 
{\par\noindent\hspace{\parindent}\textit{Proof.} } 
{\qed} 

\theoremstyle{italicheader}

\newtheorem{proposition}{Proposition}

\newtheorem{definition}{Definition}

\newtheorem{lemma}{Lemma}

\begin{document}
\renewcommand{\arraystretch}{1.3}


\title{Two fundamental solutions to the rigid Kochen-Specker set problem and the solution to the minimal Kochen-Specker set problem under one assumption}


\author{Stefan Trandafir}
\affiliation{Departamento de F\'{\i}sica Aplicada II, Universidad de Sevilla, E-41012 Sevilla, Spain}

\author{Ad\'an~Cabello}
\email{adan@us.es}
\affiliation{Departamento de F\'{\i}sica Aplicada II, Universidad de Sevilla, E-41012 Sevilla,
Spain}
\affiliation{Instituto Carlos~I de F\'{\i}sica Te\'orica y Computacional, Universidad de
Sevilla, E-41012 Sevilla, Spain}



\begin{abstract}
Recent results show that Kochen-Specker (KS) sets of observables are fundamental to quantum information, computation, and foundations beyond previous expectations. Among KS sets, those that are unique up to unitary transformations (i.e., ``rigid'') are especially important. The problem is that we do not know any rigid KS set in $\mathbb{C}^3$, the smallest quantum system that allows for KS sets. Moreover, none of the existing methods for constructing KS sets leads to rigid KS sets in $\mathbb{C}^3$.
Here, we show that two fundamental structures of quantum theory define two rigid KS sets. One of these structures is the super-symmetric informationally complete positive-operator-valued measure. The other is the minimal state-independent contextuality set. The second construction provides a clue to solve the minimal KS problem, the most important open problem in this field. 
We prove that there is no KS set of 30 elements that can be obtained from the minimal state-independent contextuality set by completing bases and adding elements that are orthogonal to two previous elements. We conjecture that 31 is the solution to the minimal KS set problem.
\end{abstract}


\maketitle


{\em Introduction---}Kochen-Specker (KS) sets \cite{Kochen:1967JMM} have been traditionally used to prove the impossibility of noncontextual hidden-variable models of quantum theory \cite{Kochen:1967JMM}, to produce bipartite perfect quantum strategies that allow two uncommunicated players to win every round of a nonlocal game \cite{Stairs:1983PS,HR83,Cabello:2001PRLa,Cabello:2001PRLb,CinelliPRL2005,YangPRL2005,Aolita:2012PRA,Cabello:2021PRL,Xu:2022PRL,Sheng:2025PRL}, and to experimentally test nature's state-independent contextuality \cite{Cabello:2008PRL,Badziag:2009PRL,Kirchmair:2009NAT,Amselem:2009PRL,D'Ambrosio:2013PRX}. A {\em KS set} is a finite set of rank-one observables $\mathcal{V}$ in a Hilbert space ${\cal H}=\mathbb{C}^d$ of finite dimension $d \ge 3$, which does not admit an assignment $f: \mathcal{V} \rightarrow \{0,1\}$ satisfying $f(u) + f(v) \leq 1$ for $u, v \in \mathcal{V}$ orthogonal, and $\sum_{u \in b} f(u) = 1$ for every orthonormal basis $b \subseteq \mathcal{V}$.

Yu and Oh \cite{Yu:2012PRL} showed that KS sets are {\em not} needed for quantum state-independent contextuality, as simpler sets, called state-independent contextuality (SI-C) sets \cite{Yu:2012PRL,Bengtsson:2012PLA}, are sufficient to prove SI-C. A {\em SI-C set} is a finite set of rank-one observables $\mathcal{V}$ in $\mathbb{C}^d$ of finite dimension $d \ge 3$, for which there is a noncontextuality inequality \cite{Cabello:2008PRL,Kleinmann:2012PRL} that is violated by any quantum state when the measurements are taken from the SI-C set. Every KS set is a SI-C set, but not every SI-C set is a KS set \cite{Yu:2012PRL,Bengtsson:2012PLA}.
It has been proven that, in quantum theory, the SI-C set with the smallest number of elements has $13$ elements and occurs in $\mathbb{C}^3$ \cite{Cabello:2016JPA}. In contrast, the simplest KS set in $\mathbb{C}^3$ {\em known} has $31$ rank-one observables and it has been proven that no KS sets exist in $\mathbb{C}^3$ with less than $24$ rank-one observables \cite{Kirchweger:2023,Li:2024}. In arbitrary $d$, it has been proven \cite{Xu:2020PRL} that the simplest KS set has $18$ observables and occurs in $\mathbb{C}^4$ \cite{Cabello:1996PLA}.

The simplicity of SI-C sets compared to the complexity of KS sets might lead one to think that KS sets are just a historical curiosity after the result of Yu and Oh \cite{Yu:2012PRL}. However, recent results \cite{Xu:2024PRL,Liu:2023XXX,Cabello:2025PRL} have shown that KS sets are important in quantum information, quantum computation, and quantum foundations in their own right. First, because KS sets are {\em necessary} for bipartite perfect quantum strategies \cite{Cabello:2025PRL}. Second, because a quantum correlation $p = \{p(a,b \vert x, y)\}$, where $x$ and $y$ are Alice's and Bob's settings, and $a$ and $b$ are Alice's and Bob's outcomes, is in a face of the nonsignaling polytope with no local points \cite{Goh:2018PRA} if, {\em and only if,} $p$ defines a KS set \cite{Liu:2023XXX,Cabello:2025PRL}. Third, because $p$ has maximum nonlocal content \cite{Elitzur:1992PLA} if, {\em and only if,} $p$ defines a KS set \cite{Liu:2023XXX,Cabello:2025PRL}. Fourth, because there is a bipartite ``all-versus-nothing'' or Greenberger-Horne-Zeilinger-like proof if, {\em and only if,} the underlying strategy defines a KS set \cite{Liu:2023XXX,Cabello:2025PRL}. Fifth because, through the above results, KS sets are related to the solution of the Tsirelson problem \cite{Ji:2021CACM} and to the proof of nonoracular quantum computational advantage in shallow circuits \cite{Bravyi:2018SCI}.

Among KS sets, ``rigid'' KS sets are particularly important. A KS set $\{\ket{{\psi_i}}\}_{i=1}^n$ in a Hilbert space ${\cal H}= \mathbb{C}^d$, with $d \ge 3$, that satisfies the orthogonality and completeness conditions given by an orthogonality graph $G$ (in which vertices represent projectors and edges indicate which ones are mutually orthogonal), is {\em rigid} if any other set of projectors $\{\Pi_i\}_{i=1}^n$ (not necessarily of rank-one) in an arbitrary (but finite) dimensional Hilbert space $\mathbb{C}^D$, with $D \ge d$, that satisfies the same orthogonality and completeness relations given by $G$, can be related to the reference KS set by a unitary operator $U$ such that, for all $i$, 
\begin{equation} \label{UnRep} 
U \Pi_i U^\dagger = \ket{{\psi_i}}\!\bra{{\psi_i}} \otimes \id, 
\end{equation}
where $\id$ is the identity operator.
Sixth, a complete KS set can be Bell self-tested \cite{Yao_self,Supic:2020Q} if, {\em and only if,} the KS set is rigid \cite{Xu:2024PRL}. Seventh, a KS set can be certified using {\em any} state of full rank if, {\em and only if,} the KS set is rigid \cite{Xu:2024PRL}. Eighth, the only known way for self-testing supersinglets of $d$ particles of $d$ levels \cite{Cabello:2002PRL,Cabello:2003JMP,PhysRevA.106.033314} is by using rigid KS sets \cite{Saha:2025XXX}. In fact, following the strategy in \cite{Saha:2025XXX}, rigid KS sets allow us to Bell self-test any $N$-partite state in which, for every bipartition with $N-1$ parties on one partition and one party on the other partition, the $N-1$ parties can predict with certainty the value of all the observables of the KS set corresponding to the other party.


{\em The rigid KS problem---}The problem is that, in ${\cal H}=\mathbb{C}^3$, the smallest quantum system (Hilbert space) where KS sets exist, we do not know any rigid KS set. The original 117-observable KS set~\cite{Kochen:1967JMM}, used in the cover of books \cite{Manin:1981,Redhead:1987}, is not rigid (see Appendix~A). The KS set that has replaced it in the cover of books \cite{Halvorson:2011} and is used in the free-will theorem \cite{CK06,CK09,conway_kochen_2011}, namely, the 33-observable KS set introduced by Peres \cite{Peres:1991JPA}, hereafter called Peres-33, which is the KS set in $\mathbb{C}^3$ with the smallest number of bases known, is not rigid, as shown in \cite{gould2010isomorphism, bengtsson2012gleason,Xu:2024PRL}: it has the same orthogonality graph as a KS set introduced by Penrose \cite{Penrose:2000}, hereafter called Penrose-33, that is not equivalent under unitary transformations. 

Moreover, none of the known methods to construct KS sets \cite{Peres:1993,Cabello:1995JPA,Zimba:1993SHPS,Cabello:1996JPA,Cabello:2005PLA,MatsunoJPA2007,Cabello:2018PRA,Ramanathan:2020Q,Cabello:2021PRL} can produce rigid KS sets in ${\cal H}=\mathbb{C}^3$ (see Appendix~B). 

In this paper, we solve the rigid KS set problem by noticing that two fundamental structures in quantum theory, namely the super-symmetric informationally complete positive-operator-valued measure (super-SIC-POVM; hereafter super SIC) \cite{Zhu:2015AP} and the minimal state-independent contextuality set (hereafter minimal SI-C set; not to be confused with SIC) \cite{Yu:2012PRL,Cabello:2016JPA}, each determines a rigid KS set. Our approach also solves a problem left open in \cite{Xu:2024PRL}, namely, whether or not the KS set with the minimum number of observables {\em known} in $\mathbb{C}^3$ \cite{Peres:1993} is rigid, and provides an unexpected insight on the problem of what is the minimal KS set in ${\cal H}=\mathbb{C}^3$ \cite{Jost:1976,Peres:1988,Peres:1993,Pavicic:2005JPA,Arends:2011,Uijlen:2016,Kirchweger:2023,Li:2024,williams2024maximalnonkochenspeckersetslower}.


{\em Rigid KS set defined by the super SIC---}SIC-POVMs (hereafter just SICs) \cite{Renes2004JMP} are fundamental for many reasons \cite{Fuchs:2017A,DeBrota:2020PRR}. However, among all SICs, the SIC in $\mathbb{C}^3$ is special: $\mathbb{C}^3$ is one of the three cases in which the symmetry groups act transitively on pairs of SIC elements \cite{Zhu:2015AP}. These SICs are covariant with respect to Heisenberg-Weyl groups and their symmetry groups are subgroups of Clifford groups that act transitively on pairs of SIC projectors. However, only in $\mathbb{C}^3$ the SIC is covariant with respect to the Clifford group. For this reason, the SIC in $\mathbb{C}^3$ is called the ``super-symmetric informationally complete measurement'' \cite{Zhu:2015AP}. 


\begin{figure}
 \vspace{-1.2cm}
 \centering
 \includegraphics[width=0.795\linewidth]{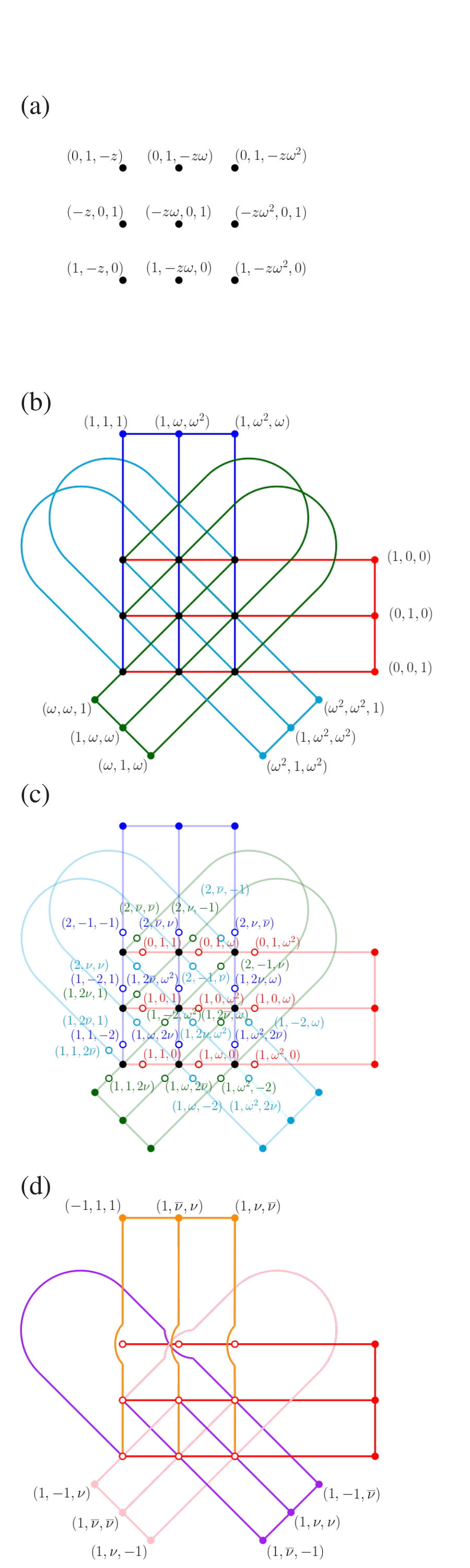}
 \caption{(a)-(d) illustrate steps (I)-(IV), respectively, of the construction of the rigid KS set defined by the super SIC; see the main text for details. In (b) and (d), the dots in a line of the same color as the dot are orthogonal to the latter dot. $z=1$, $\omega = e^{2i\pi/3}$, $\nu = e^{i\pi/3}$, and $\overline{\nu} = e^{-i\pi/3}$.}
 \label{fig:main-figure}
\end{figure}


Let us now show how the SIC in $\mathbb{C}^3$ defines a rigid KS set in $\mathbb{C}^3$. The construction is as follows:

(I) Every SIC in $\mathbb{C}^3$ is unitarily equivalent to a SIC of the form in Fig.~\ref{fig:main-figure} (a), where $\omega$ is a third root of unity and $z$ is an arbitrary phase factor \cite{Szollosi:2014XXX,Hughston:2016AM}. Hereafter, we will take $\omega=e^{\frac{i 2 \pi}{3}}$ and $z = 1$. This corresponds to the so-called {\em Hesse SIC} \cite{Bengtsson:2012PLA,Zhu:2015AP,Fuchs:2017A}, which is rigid.

(II) Wootters \cite{Wootters:2006quantum} pointed out that the nine SIC elements of the Hesse SIC determine four mutually unbiased bases (MUBs). Each MUB element is orthogonal to three elements of the Hesse SIC. Wootters' construction is shown in Fig.~\ref{fig:main-figure} (b).
The resulting $9+12$-element set is called BBC-21 \cite{Xu:2024PRL} and is a SI-C set but not a KS set \cite{Bengtsson:2012PLA}. As it is clear from the way BBC-21 is constructed, BBC-21 is rigid. An independent proof of the rigidity of BBC-21 can be found in \cite{Xu:2024PRL}. 

(III) If we start from BBC-21, every orthogonal pair (SIC element, MUB element) determines a new element: the one that is orthogonal to both of them. Since there are $36$ pairs, each yielding a unique new element, there are $36$ new elements, which are illustrated in Fig.~\ref{fig:main-figure} (c). By construction, the set with the $9+12$ old and the $36$ new elements is rigid.
Now is when we make a crucial observation: The $36$ new elements can be partitioned into four disjoint SICs. In Fig.~\ref{fig:main-figure} (c) we assign a different color to each of the four new SICs. The $9$ red dots (with white inside) define a SIC, and similarly for the green, blue, and cyan dots.

(IV) Each of the four new SICs determine three new MUBs. 
These three MUBs form a complete set of MUBs with one of the Hesse MUBs (a different one for each of the four new SICs).
The construction of the three MUBs associated to the ``red'' SIC is illustrated in Fig.~\ref{fig:main-figure} (d). The constructions for the other three SICs are similar (see Appendix~C).

In total, we obtain a set of $9 + 12 + 36 + 3 \times 12 = 93$ elements, which is rigid by construction. We can remove the $12$ elements that have degree two, as they do not constrain the possible non-contextual assignments of the other elements. The resulting $81$-element set is a KS set as it can be checked with the aid of a simple program \cite{Peres:1993} or an Integer Linear program \cite{Salt:2023}. An analytic proof can be found in Appendix~D. 


\begin{figure}[t]
 \centering \includegraphics[width=0.371\textwidth]{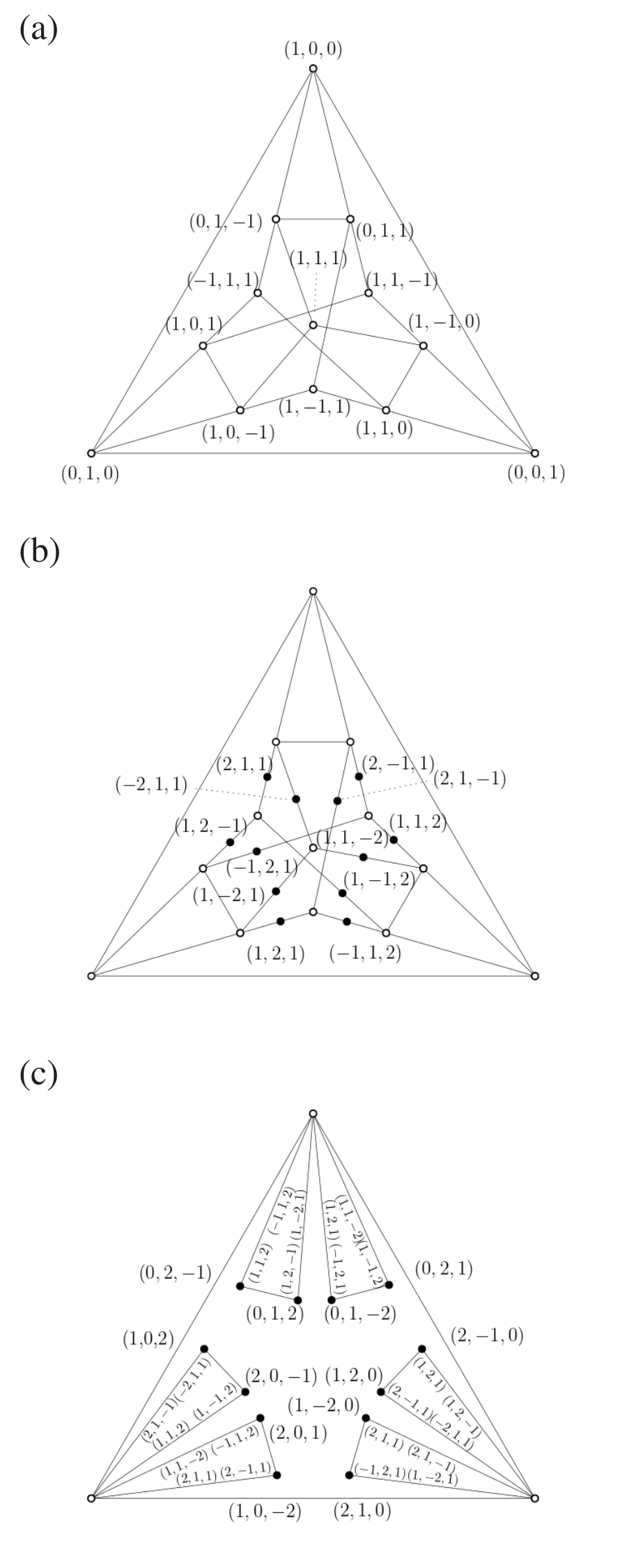}
 \caption{(a)-(c) illustrate steps (i)-(iii), respectively, of the construction of the rigid KS set defined by the minimal SI-C set; see the main text for details. Dots in the same line or in the same triangle represent mutually orthogonal vectors. In (c), the edges connecting the black vertices with a vector of the canonical basis are labeled by the vectors added in (b). For example, $(0,1,2)$ is the unique vector orthogonal to $(1,0,0)$ and $(1,2-1)$, and also the unique vector orthogonal to $(1,0,0)$ and $(1,-2,1)$.}
 \label{fig:yu-oh}
\end{figure}


{\em Rigid KS set defined by the minimal SI-C set---}The minimal SI-C in every Hilbert space is the $13$-element set found by Yu and Oh \cite{Yu:2012PRL}, which is illustrated in Fig.~\ref{fig:yu-oh} (a). As proven in \cite{Yu:2012PRL,bengtsson2012gleason,Xu:2024PRL}
this set is rigid.

Let us now show how the minimal SI-C defines a rigid KS set in $\mathbb{C}^3$. The construction is as follows:

(i) We start with the minimal SI-C set in Fig. \ref{fig:yu-oh} (a).

(ii) For every orthogonal pair that is not in a basis, we add the vector that is orthogonal. This adds $12$ vectors represented by black vertices in Fig. \ref{fig:yu-oh} (b). 

(iii) For each pair consisting of a vector of the canonical basis and a vector added in step (ii), we add the orthogonal vector. This adds $12$ vectors represented in black in Fig.~\ref{fig:yu-oh} (c). 

The resulting set, consisting of the $13$ vectors in (i), plus the $12$ in (ii), plus the $12$ in (iii) is, by construction, rigid. However, the resulting set is not new: it is a set found by Conway and Kochen in the 1990's and communicated to Peres \cite{Peres:1991JPA,Peres:1993} (see Appendix~E), which is known to be a KS set. Hereafter, we will refer to this set as CK-37. 


{\em Critical rigid KS sets---}Zimba and Penrose define a KS set to be {\em critical} if we cannot remove any of its elements without losing the property of being a KS set \cite{Zimba:1993SHPS}. Neither the rigid KS set associated to the super SIC (henceforth KS-81) nor CK-37 are critical. Therefore, two crucial questions are what are the critical KS sets contained in KS-81 and CK-37 and whether these subsets are, themselves, rigid.

The smallest critical subset of $KS-81$ that we have found has $55$ elements (see Appendix~G).
CK-37 has two critical KS sets. Both were identified by Conway and Kochen (see Appendix~E), so we will refer to them as CK-33 and and CK-31, as they have 33 and 31 elements, respectively. CK-33 was previously found by Sch{\"u}tte \cite{Bub:1996FP} and is different than the $33$-element set of Peres \cite{Peres:1991JPA} (which has the same orthogonality graph as the $33$-element set of Penrose \cite{Penrose:2000}). There are three equivalent (up to unitary transformations) versions of CK-33 (depending of which four vectors we remove from CK-37). There are six equivalent (up to unitary transformations) versions of CK-31. One of them was reported by Peres \cite{Peres:1993} and is the KS set in $\mathbb{C}^3$ with the smallest number of elements {\em known}. 

Our construction of the rigid KS set from the minimal SI-C set allows us to prove the following.


{\em Theorem 1.} CK-33 is rigid.

\begin{proof}
 CK-33 can be obtained, e.g., by removing $(0,2,1)$, $(0,1,-2)$, $(0,2,-1)$, and $(0,1,2)$ from CK-37. These four vectors correspond to the upper four black dots in Fig.~\ref{fig:yu-oh} (c). $(0,2,1)$ is the unique vector orthogonal to $(1,0,0)$, $(1,2,-1)$, and $(1,-2,1)$. Therefore, we can remove it without compromising the rigidity that existed in CK-37. A similar argument explains why removing $(0,1,-2)$, $(0,2,-1)$, and $(0,1,2)$ do not compromise rigidity. Fig.~\ref{fig:yu-oh} (c) also makes clear why there are exactly three equivalent versions of CK-33 in CK-37.
\end{proof}


{\em Theorem 2.} CK-31 is rigid.

\begin{proof}
CK-31 can be obtained by, e.g., removing $(2,1,1)$, $(2,1,0)$, $(2,1,-1)$, $(-1,2,1)$, $(1,-2,0)$, and $(1,-2,1)$ from CK-37. These six vectors are {\em all} the vectors in the lower right small triangle in Fig.~\ref{fig:yu-oh} (c).
 
We can prove the rigidity of CK-31 by following a similar procedure as for CK-37. Steps (I) and (II) are, in fact, the same as (i) and (ii).
 
(III) We remove the four vectors $(2,1,1)$, $(2,1,-1)$, $(-1,2,1)$, $(1,-2,0)$, which were each obtained in (II). The resulting set of 21 vectors is rigid.

(IV) For each pair consisting of a vector of the canonical basis and one of the eight vectors not in minimal SI-C set, we add the orthogonal vector. This process adds 10 new vectors (i.e., those corresponding to black dots of Fig.~\ref{fig:yu-oh} (c) that are not in the lower right small triangle). The resulting set is CK-31.

Rigidity follows from the fact that we started with a rigid set and each new vector added was orthogonal to two added in the previous step.
 It is then clear that we can remove the six without compromising the rigidity. Fig.~\ref{fig:yu-oh} (c) also makes clear why there are exactly six equivalent versions of CK-31 in CK-37.
\end{proof}


{\em The minimal KS set problem---}We have been looking for the minimal KS set in $\mathbb{C}^3$ for decades using all kinds of methods \cite{Jost:1976,Peres:1988,Peres:1993,Pavicic:2005JPA,Arends:2011,Uijlen:2016,Kirchweger:2023,Li:2024,williams2024maximalnonkochenspeckersetslower}, but we still do not have the answer. It has only been proven that it has to have more than 23 elements \cite{Kirchweger:2023,Li:2024} and, at most 31 \cite{Peres:1993}. However, the proof that CK-31 is rigid and, specially, that CK-31 is {\em determined} by the minimal SI-C set changes the traditional (brute-force) approach and strongly suggests that the answer to the minimal KS set is $31$.

The argument is as follows. The minimal KS set {\em must} be a SI-C set. It has been proven \cite{Cabello:2017} that the minimal SIC set (in any dimension) is the one in Fig. \ref{fig:yu-oh} (a). So far, we have proven that the minimal KS set {\em known} \cite{Peres:1993} is {\em determined} by the minimal SI-C set in the sense that it follows from, first, completing bases and, then, adding {\em some} new elements that are orthogonal to two existing elements, and, finally, removing (without compromising rigidity) unnecessary elements. That is, it is determined by a rigid SI-C set by adding elements using two operations that preserve rigidity.
Interestingly, we can prove an even stronger result.

{\em Theorem 3.} There is no KS set of 30 (or less) elements that is obtained from the minimal SI-C set by, first, completing bases and, then adding {\em all possible} new elements that are orthogonal to, at least, two previous elements, and, finally, removing some elements.

\begin{proof}
We start with the minimal SI-C set that has 13 elements. After completing all incomplete bases we end up with the 25-element SI-C set in Fig.~\ref{fig:yu-oh} (b). 
Now we add all elements that are orthogonal to, at least, two of the $25$ existing elements. 
There are exactly $72$ new vectors that satisfy this requirement. The problem is finding the smallest KS set within this set of $97$ elements. The solutions requires noticing that any element in a KS set that is not in a complete basis can be removed and the resulting set is still a KS set. Therefore, we can remove all the elements that are not in complete basis. This leads us with a subset of $37$ elements, which is a KS set. The smallest KS set in it has $31$ elements.
\end{proof}

If we assume that the minimal KS set must contain the minimal SI-C set and must be rigid, then Theorem~3 is pointing in the direction that $31$ is the minimum KS set in $\mathbb{C}^3$. It is not a proof because we cannot discard that a smaller KS set would appear when we consider a second round of completing basis. Such a second round leads to a set of $1741$ elements. After removing those elements that are in a unique orthogonal basis, and then removing elements with degree $2$, we obtain a set of $805$ elements. However, finding whether there is a $30$-element KS subset inside this set is out of our computing capabilities. Still, Theorem~3 leads us to formulate the following.

{\em Conjecture 1.} There is no rigid KS set of 30 (or less) elements that contains the minimal SI-C set.

Therefore, under the assumption that the minimum KS set is rigid and contains the minimal complete SI-C set, our conjecture implies that there is no smaller KS set than CK-31. 
The requirement of rigidity is natural in two senses. On the one hand, to convert a non-KS set into a KS set, we need the added vectors to be orthogonal to, at least, three other vectors of the set. Asking that two of the added vectors be orthogonal to two of the minimal complete SI-C set seems a weak requirement. On the other hand, asking a fundamental quantum object such as the minimal KS set to be rigid seems natural.

In principle, there is the possibility that the minimal KS set does {\em not} contain the minimal SI-C set. However, it is very unlikely for two reasons. First, all known small KS sets contain the minimal SI-C set: CK-37, CK-33, CK-31, Peres-33, and Penrose-33. Second, the next {\em known} SI-C set which does not contains the minimal SI-C set is BBC-21, which has 21 elements and, after completion is the $21+36$ element set in Fig.~\ref{fig:main-figure} (b), which has too many elements to be the minimal KS set.
Therefore, Conjecture~1 (supported by Theorem~3) strongly suggests that CK-31 is the minimal KS set in $\mathbb{C}^3$ allowed by quantum theory.


{\em Conclusions and open problems---}The last years have completely changed our perspective on why KS sets are important. We have proven that they have to be in {\em every} bipartite perfect quantum strategy, in {\em every} bipartite fully nonlocal quantum correlation, in {\em every} bipartite quantum correlation that ``touch'' the nonsignaling bound (specifically, a face of the nonsignaling polytope which do not have local points). Moreover, several fundamental recent results on quantum computation and quantum foundations rely on these correlations and, therefore, rely, ultimately, on KS sets. In addition, several recent applications demand rigid KS sets.

Here, we have solved two problems (but one of them only partially). On the one hand, we have solved the ``rigid KS set problem'' by identifying five rigid KS sets in $\mathbb{C}^3$: two of them come from the super SIC, three of them come from the minimal SI-C and were known (although their authors, Conway and Kochen-- never published them and never used them because they were less symmetrical than other alternatives, see Appendix~F). On the other hand, in the process of solving the rigid KS set problem, we have found a strong connection between this problem and the main open problem in the field, namely, the ``minimum KS problem.'' Thanks to this connection, we have been able to prove that there is no KS set with $30$ elements containing the minimal complete SI-C set and elements that are orthogonal to two elements of the minimal SI-C set. This result strongly suggests that {\em the} minimal KS set in quantum theory has 31 observables. This result is not only crucial in foundations of quantum theory but, in light of the recently found key roles that KS sets play in quantum information and computation (see the Introduction), it is important in a broad sense. 

SI-C sets of observables are important in contextuality and foundations of physics. However, we know little about them. We know that, in $\mathbb{C}^3$, the minimal SI-C set and BBC-21 are the first two members of a family \cite{Xu:2015PLA}. We know the minimal SI-C sets in every $\mathbb{C}^d$ \cite{Cabello:2018PRA}, and we know how to check whether a set is a SI-C \cite{CKB2015}. However, we do not have systematic ways to obtain SI-C sets. In that respect, our results pave the way for the search for new SI-C sets, since they provide two different ways to, potentially, produce new examples. On the one hand, the construction associated to the super SIC, could be extended to some other special dimensions. Further research is needed in this direction. On the other hand, comparing the two constructions allows us to identify a pattern: Step 1: take a (possibly symmetric) POVM. Step 2: calculate all vectors orthogonal to two of them.
Step 3: calculate all vectors orthogonal to (at least) two of the vectors from the two previous steps. Step 4: check for SI-C. However, this will probably produce SI-C sets with more than 21 elements and, as a consequence, KS sets with more than 30 elements. 


{\em Acknowledgments---}We thank Ingemar Bengtsson, Emmanuel Briand, Chris Fuchs, and Zhengyu Li for helpful discussions, references, and comments.
This work was supported by the EU-funded project \href{10.3030/101070558}{FoQaCiA}, the \href{10.13039/501100011033}{MCINN/AEI} (Project No.\ PID2020-113738GB-I00), and the Wallenberg Center for Quantum Technology (WACQT).


\newpage
\mbox{}
\newpage
\onecolumngrid
\appendix


\section{Appendix A: The 117-observable KS set is not rigid}

\setcounter{section}{\value{section}}


The orthogonality graph of the 117-observable KS set of Ref.~\cite{Kochen:1967JMM} is shown in Fig.~\ref{Fig:KS-117}. The reason why this orthogonality graph corresponds to a KS set in $\mathbb{C}^3$ is explained in the caption of Fig.~\ref{Fig:KS-117}.


\begin{figure}[t!]
\centering
\includegraphics[width=0.55\linewidth]{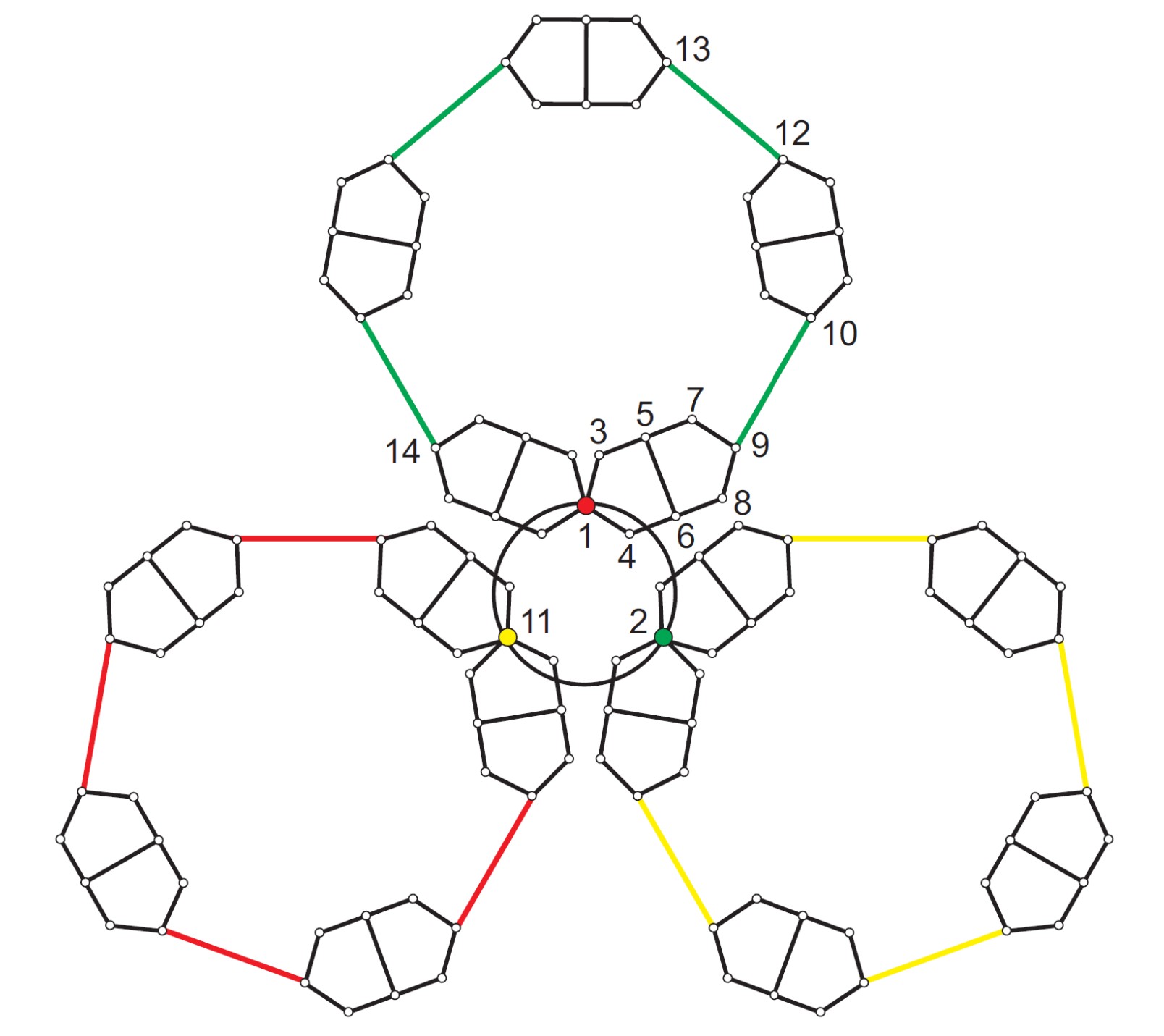}
\caption{Orthogonality graph of the 117 rank-one projectors in ${\cal H}=\mathbb{C}^3$ of the KS set $\mathcal{V}$ in Ref.~\cite{Kochen:1967JMM}. Nodes in the same straight line or circumference represent mutually orthogonal projectors. The red node is orthogonal to all nodes connected by a red edge. Similarly for the green and yellow nodes. That $\mathcal{V}$ does not admit a KS assignment $f: \mathcal{V} \rightarrow \{0,1\}$ satisfying $f(u) + f(v) \leq 1$ for $u, v \in \mathcal{V}$ orthogonal, and $\sum_{u \in b} f(u) = 1$ for every orthonormal basis $b \in \mathcal{V}$ can be seen as follows. One of the nodes $1$, $2$, and $11$ has to be assigned value $1$. Without loss of generality, the symmetry of the graph allows us to assume that it is node $1$. That is, we assume that $f(1)=1$. Then, $f(9)=0$ because of the subset $\{1,3,4,5,6,7,8,9\}$.
Then, since nodes $2$, $9$, and $10$ are mutually orthogonal and node $2$ is connected to node $1$, then $f(10)=1$. Applying the same argument, $f(12)=0$ and $f(13)=1$, since $\{2, 12, 13\}$ form an orthogonal basis. Repeating it again twice, $f(14)=1$. However, nodes $1$ and $14$ cannot be both assigned value $1$. This proves that $\mathcal{V}$ is a KS set. The figure is taken from \cite{Budroni:2022RMP}.}
\label{Fig:KS-117}
\end{figure}


The proof that the set is not rigid is as follows. Notice that Fig.~\ref{Fig:KS-117} contains $15$ copies of a $10$-node structure (see nodes $1$ to $10$ in Fig.~\ref{Fig:KS-117}). Without loss of generality, we can assume that nodes $1$ and $2$ correspond to the vectors
\begin{align}
 1&=(1,0,0),\\
 2&=(0,0,1).
\end{align}
Then, we can chose the vectors corresponding to the other eight nodes as follows:
\begin{align}
 3&=(0,\cos \alpha, \sin \alpha),\\
 4&=(0,\cos \beta, \sin \beta),\\
 5&=(\tan \phi \csc \alpha,-\sin \alpha, \cos \alpha),\\ 
 6&=(\tan \phi \csc \beta,-\sin \beta, \cos \beta),\\ 
 7&=(\cot \phi, 1, -\cot \alpha),\\
 8&=(\cot\phi, 1, -\cot \beta),\\ 
 9&=(\sin \phi,-\cos \phi,0),\\ 
 10&=(\cos \phi,\sin \phi,0),
\end{align}
with $\alpha \neq \beta$ and $\beta \neq \frac{p \pi}{2}$, with $p$ integer. Since nodes $5$ and $6$ are orthogonal, then,
\begin{equation} \label{eq231}
 \sin \alpha \sin \beta \cos (\alpha-\beta)=-\tan^2 \phi.
\end{equation}
Since the left-hand side of Eq.~\eqref{eq231} is in $[-\frac{1}{8}, 1]$, then
\begin{equation}
 |\phi| \le \arctan\frac{1}{\sqrt{8}}.
\end{equation}
Therefore, there is plenty of room to chose $\phi$ (and then $\alpha$ and $\beta$) for most of the 10-node structures in Fig.~\ref{Fig:KS-117}. Consequently, the 117-observable KS set of Ref.~\cite{Kochen:1967JMM} is not rigid.


\section{Appendix B: None of the known methods to construct KS sets produce rigid KS sets in $\mathbb{C}^3$}


The methods to construct KS sets are, essentially, of two types. One type groups those methods that produce a KS set in $\mathbb{C}^D$
starting from a KS set in $\mathbb{C}^d$, with $d < D$ \cite{Peres:1993,Zimba:1993SHPS,Cabello:1996JPA,Cabello:2005PLA,MatsunoJPA2007}. These methods cannot produce KS sets in $\mathbb{C}^3$, since KS sets are impossible in $\mathbb{C}^2$ \cite{Kochen:1967JMM}.

The other type groups those methods that concatenate basic structures such as the $10$-node structure made by nodes $1$ to $10$ in Fig.~\ref{Fig:KS-117} to produce a KS set \cite{Kochen:1967JMM,Cabello:1995JPA,Cabello:2018PRA,Ramanathan:2020Q,Cabello:2021PRL}. There is an infinite number of these structures in any $\mathbb{C}^d$, with $d \ge 3$ \cite{Cabello:1995JPA,Cabello:2018PRA,Ramanathan:2020Q}. However, none of the minimal ones are rigid \cite{Cabello:2018PRA}. Moreover, as these structures become more complex, they also become less rigid \cite{Cabello:1995JPA}. Consequently, every KS constructed by {\em concatenating} these structures will not be rigid.

\newpage


\begin{figure}[h]
\centering
\includegraphics[width=0.8\linewidth]{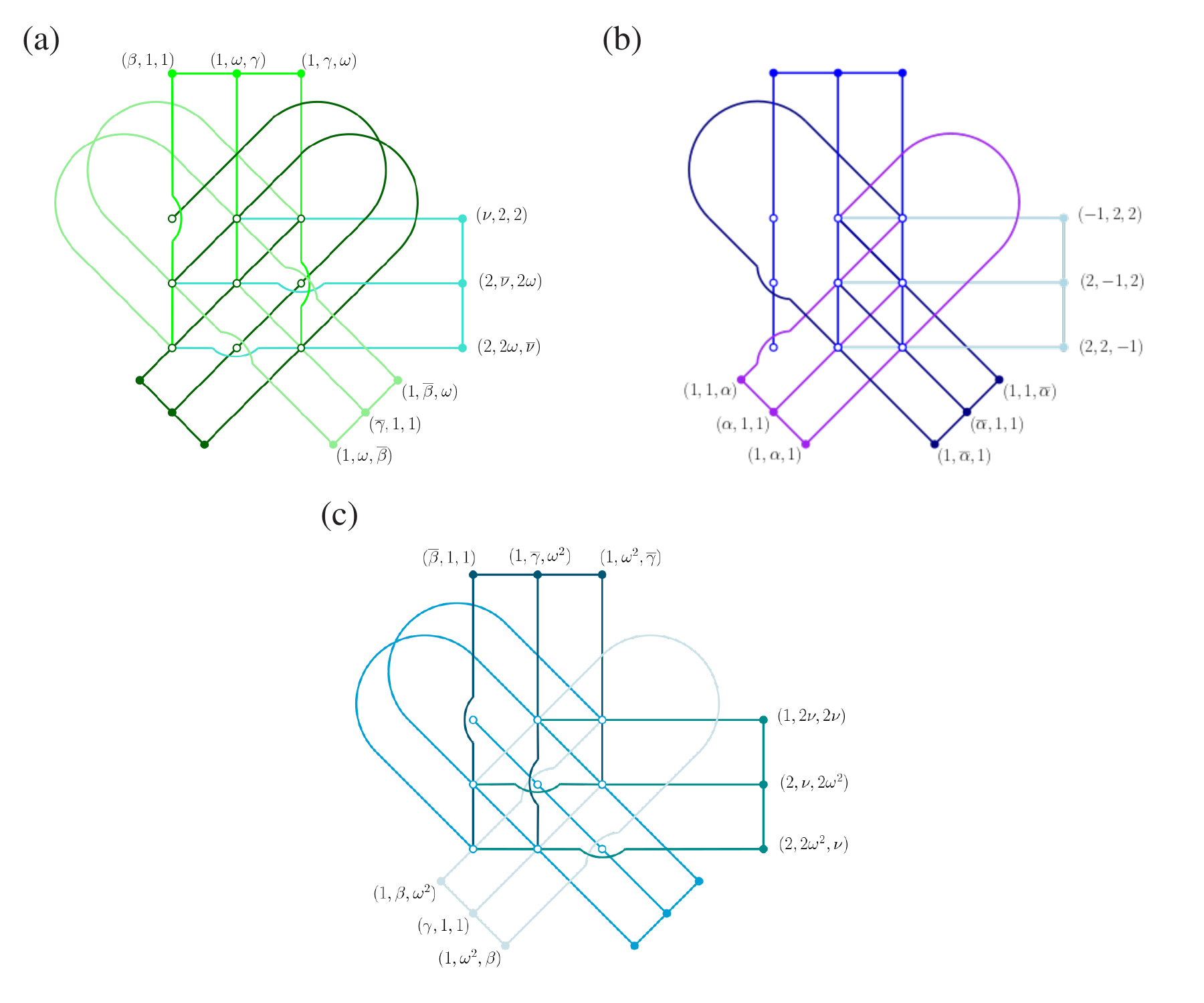}
\caption{(a) Construction of three MUBs associated to the ``green'' SIC. (b) Construction of the three MUBs associated to the ``blue'' SIC. (c) Construction of the three MUBs associated to the ``cyan'' SIC. $\omega = \frac{1}{2}(-1+i\sqrt{3})$, $\nu = \frac{1}{2}(1+i\sqrt{3})$, $\alpha = \frac{1}{2}(-1+i3\sqrt{3})$, $\beta = -2+i\sqrt{3}$, $\gamma=\frac{1}{2}(5+i\sqrt{3})$, and $\overline{x}$ denotes the conjugate of $x$.}
\label{fig:the-mub-construction-business}
\end{figure}


\section{Appendix C: Details on the construction of the rigid KS set associated to the super SIC}


In Fig.~1 in the main text, we described how KS-81 (i.e., the rigid KS set associated to the super SIC) is constructed. There, in Fig.\ 1 (c), we described how the three new MUBs associated to the ``red'' SIC are constructed. Here, we do the same for the other three SICs. Specifically, Fig.~\ref{fig:the-mub-construction-business} (a) shows the construction of the three MUBs associated to the ``green'' SIC. Similarly, Fig.~\ref{fig:the-mub-construction-business} (b) shows the construction of the three MUBs associated to the ``blue'' SIC, and Fig.~\ref{fig:the-mub-construction-business} (c) shows the construction of the three MUBs associated to the ``cyan'' SIC.

\newpage


\begin{figure*}
\centering
\includegraphics[width=0.66\linewidth]{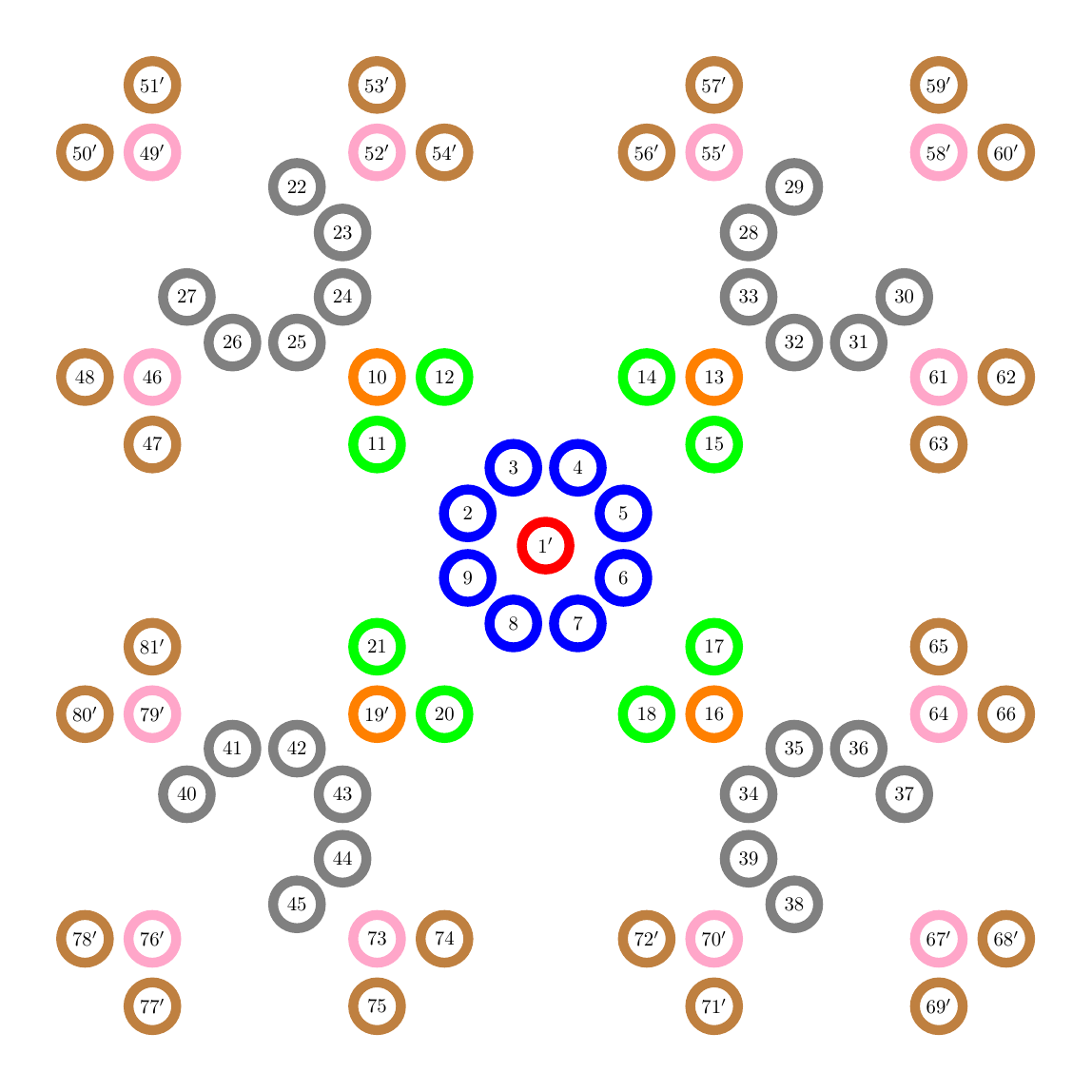}
\caption{KS-81 (i.e., the rigid 81-element KS set defined by the super SIC). Each element is represented by a circle labeled with a number. The elements are given explicitly in Table~\ref{tab:CT93}. Each of the seven colors corresponds to one of the seven orbits of the vertices of the orthogonality graph of the set under its automorphisms. Those elements containing an apostrophe do not belong to the $55$-element critical KS subset detailed in Appendix~G.}
 \label{fig:CT93}
\end{figure*}


\section{Appendix D: Detailed description of KS-81, the rigid KS set associated to the super SIC}


The $81$ vectors of KS-81 (i.e., the rigid KS set associated to the super SIC) are listed in Table~\ref{tab:CT93} and illustrated in Fig.~\ref{fig:CT93}. In Fig.~\ref{fig:CT93}, the central ``octagon'' of blue circles surrounding a red
circle represents the nine elements of a SIC. In addition, there are four sets of six gray circles. Each of them represents six elements of a SIC. The four ``triangles'' of circles surrounding each SIC represent a set of MUBs.
Together, the central SIC (that we will call the \emph{Hesse SIC}) and its associated MUBs (the \emph{Hesse MUBs}) form BBC-21 (the Hesse configuration). KS-81 contains the entire Hesse SIC, but only six of the nine elements of each of the other four SICs (hereafter called {\em non-Hesse SICs}). 

Each of the $24$ gray circles is in a different orthogonal basis together with the blue circle that is in the same position in the blue octagon, and one circle from the MUB between the gray set and the blue octagon. 

There are additional orthogonalities between the Hesse-SIC and each of the 12 non-Hesse MUBs, and between the non-Hesse SICs and their corresponding MUBs. No other orthogonalities exist. 

The orthogonality graph of KS-81 has an automorphism group $A$ of order $48$ and is isomorphic to $GL_2(\mathbb{F}_3)$ (i.e., the group of $2 \times 2$ invertible matrices with entries from the field of three elements under multiplication). The seven different colors in Fig.~ \ref{fig:CT93} correspond to the seven orbits of the vertices of the orthogonality graph under its automorphisms. 


\begin{table}[h]
 \centering
 \begin{tabular}{lccc|lccc|lccc|lccc|lccc|lccc}
 \hline \hline
 No. & $v_1$ & $v_2$ & $v_3$ & \;No. & $v_1$ & $v_2$ & $v_3$ & \;No. & $v_1$ & $v_2$ & $v_3$ & \;No. & $v_1$ & $v_2$ & $v_3$ & \;No. & $v_1$ & $v_2$ & $v_3$ & \;No. & $v_1$ & $v_2$ & $v_3$ \\ \hline
$1'$ & $0$ & $1 $ & $-1$ & \;$15$ & $\omega^2$ & $1$ & $\omega^2$ & \;$29$ & $1$ & $1$ & $2\overline{\nu}$ & \;$43$ & $1$ & $\omega^2$ & $0$ & \;$57'$ & $2$ & $2\omega^2$ & $\nu$ & \;$71'$ & $1$ & $\overline{\beta}$ & $\omega$ \\
$2$ & $1$ & $0$ & $-1$ & \;$16$ & $1$ & $\omega$ & $\omega$ & \;$30$ & $1$ & $-2$ & $\omega$ & \;$44$ & $1$ & $0$ & $\omega$ & \;$58'$ & $\gamma$ & $1$ & $1$ & \;$72'$ & $1$ & $\omega$ & $\overline{\beta}$ \\
$3$ & $1$ & $-1$ & $0$ & \;$17$ & $\omega$ & $\omega$ & $1$ & \;$31$ & $1$ & $\omega$ & $-2$ & \;$45$ & $1$ & $\omega$ & $0$ & \;$59'$ & $1$ & $\omega^2$ & $\beta$ & \;$73$ & $-1$ & $1$ & $1$ \\
$4$ & $\omega$ & $0$ & $-1$ & \;$18$ & $\omega$ & $1$ & $\omega$ & \;$32$ & $2$ & $\overline{\nu}$ & $-1$ & \;$46$ & $-1$ & $2$ & $2$ & \;$60'$ & $1$ & $\beta$ & $\omega^2$ & \;$74$ & $1$ & $\nu$ & $\overline{\nu}$ \\
$5$ & $1$ & $-\omega^2$ & $0$ & \;$19'$ & $1$ & $0$ & $0$ & \;$33$ & $2$ & $-1$ & $\overline{\nu}$ & \;$47$ & $2$ & $2$ & $-1$ & \;$61$ & $\overline{\beta}$ & $1$ & $1$ & \;$75$ & $1$ & $\overline{\nu}$ & $\nu$ \\
$6$ & $\omega^2$ & $0$ & $-1$ & \;$20$ & $0$ & $1$ & $0$ & \;$34$ & $1$ & $2\nu$ & $1$ & \;$48$ & $2$ & $-1$ & $2$ & \;$62$ & $1$ & $\overline{\gamma}$ & $\omega^2$ & \;$76'$ & $1$ & $\nu$ & $\nu$ \\
$7$ & $1$ & $-\omega$ & $0$ & \;$21$ & $0$ & $0$ & $1$ & \;$35$ & $1$ & $1$ & $2\nu$ & \;$49'$ & $\alpha$ & $1$ & $1$ & \;$63$ & $1$ & $\omega^2$ & $\overline{\gamma}$ & \;$77'$ & $1$ & $\overline{\nu}$ & $-1$ \\
$8$ & $0$ & $1$ & $-\omega^2$ & \;$22$ & $1$ & $2\overline{\nu}$ & $\omega^2$ & \;$36$ & $1$ & $-2$ & $\omega^2$ & \;$50'$ & $1$ & $\alpha$ & $1$ & \;$64$ & $\nu$ & $2$ & $2$ & \;$78'$ & $1$ & $-1$ & $\overline{\nu}$ \\
$9$ & $0$ & $1$ & $-\omega$ & \;$23$ & $1$ & $\omega^2$ & $2\overline{\nu}$ & \;$37$ & $1$ & $\omega^2$ & $-2$ & \;$51'$ & $1$ & $1$ & $\alpha$ & \;$65$ & $2$ & $2\omega$ & $\overline{\nu}$ & \;$79'$ & $1$ & $\overline{\nu}$ & $\overline{\nu}$ \\
$10$ & $1$ & $1$ & $1$ & \;$24$ & $1$ & $2\nu$ & $\omega$ & \;$38$ & $2$ & $-1$ & $\nu$ & \;$52'$ & $\overline{\alpha}$ & $1$ & $1$ & \;$66$ & $2$ & $\overline{\nu}$ & $2\omega$ & \;$80'$ & $1$ & $\nu$ & $-1$ \\
$11$ & $1$ & $\omega$ & $\omega^2$ & \;$25$ & $1$ & $\omega$ & $2\nu$ & \;$39$ & $2$ & $\nu$ & $-1$ & \;$53'$ & $1$ & $1$ & $\overline{\alpha}$ & \;$67'$ & $\beta$ & $1$ & $1$ & \;$81'$ & $1$ & $-1$ & $\nu$ \\
$12$ & $1$ & $\omega^2$ & $\omega$ & \;$26$ & $2$ & $\nu$ & $\overline{\nu}$ & \;$40$ & $1$ & $0$ & $1$ & \;$54'$ & $1$ & $\overline{\alpha}$ & $1$ & \;$68'$ & $1$ & $\omega$ & $\gamma$ & \\
$13$ & $1$ & $\omega^2$ & $\omega^2$ & \;$27$ & $2$ & $\overline{\nu}$ & $\nu$ & \;$41$ & $1$ & $1$ & $0$ & \;$55'$ & $1$ & $2\nu$ & $2\nu$ & \;$69'$ & $1$ & $\gamma$ & $\omega$\; & \\
$14$ & $\omega^2$ & $\omega^2$ & $1$ & \;$28$ & $1$ & $2\overline{\nu}$ & $1$ & \;$42$ & $1$ & $0$ & $\omega^2$ & \;$56'$ & $2$ & $\nu$ & $2\omega^2$ & \;$70'$ & $\overline{\gamma}$ & $1$ & $1$ \\
 \hline \hline
 \end{tabular}
\caption{KS-81, the rigid KS set defined by the super SIC. Each element is represented by a vector $(v_1,v_2,v_3)$. $\omega = \frac{1}{2}(-1+i\sqrt{3})$, $\nu = \frac{1}{2}(1+i\sqrt{3})$, $\alpha = \frac{1}{2}(-1+i3\sqrt{3})$, $\beta = -2+i\sqrt{3}$, $\gamma=\frac{1}{2}(5+i\sqrt{3})$, and $\overline{x}$ denotes the conjugate of $x$. Numbering matches that of Fig.~\ref{fig:CT93}. The elements marked with an apostrophe do not appear in the 55-element critical KS subset detailed in Appendix~G.}
 \label{tab:CT93}
\end{table}


If one ignores the non-basis orthogonalities, then (i) there are four orbits: the Hesse SIC, the Hesse MUBs, the non-Hesse SICs, the non-Hesse MUBs, and (ii) there is a KS assignment (i.e., one needs to take the non-basis orthogonalities into account to see that KS-81 is, indeed, a KS set).


\section{Appendix E: Proof that KS-81 is a KS set}


Here, we analytically prove that KS-81 (i.e., the rigid KS set associated to the super SIC) is, in fact, a KS set. The proof will use Fig.~\ref{fig:CT93} and is as follows. Any KS assignment of $0$'s and $1$'s to the elements of KS-81 must assign $1$ to, at least, one element of each of the $16$ MUBs. Each of the elements of the $12$ non-Hesse MUBs (the pink and brown circles in Fig.~\ref{fig:CT93}) is in a single orthogonal basis. The green MUB elements are each in four orthogonal bases, and the orange MUB elements are each in one orthogonal basis.
Each of the $24$ elements of non-Hesse SICs (the grey circles) are in one orthogonal basis, and the blue circles are in three orthogonal bases, while the red circle is in zero. We call the four MUBs surrounding a non-Hesse SIC its corresponding MUBs.

\begin{lemma} \label{lem:outer-limits}
Choose any non-Hesse SIC and its corresponding MUBs, and consider the resulting set of elements. All KS assignments of this set can have at most two $1$'s assigned to the non-Hesse SIC elements.
\end{lemma}

\begin{proof}
The orthogonalities between any of the non-Hesse SICs and their four corresponding MUBs are described in Fig.~\ref{fig:cornerMUBs+outerSICs}. 
Any KS assignment must assign $1$ to exactly one element from each of the four MUBs. The MUB represented by the circles has two of its three elements pictured in Fig.~\ref{fig:cornerMUBs+outerSICs} and the other MUBs have all three elements pictured. 
From Fig.~\ref{fig:cornerMUBs+outerSICs}, one may see that any choice of three distinct markings (e.g., top left circle, top cross, and central square in Fig.~\ref{fig:cornerMUBs+outerSICs})
leaves at most two gray curves that do not touch any of the chosen markers (in our example, the right-most vertical curve and bottom horizontal curve in Fig.~\ref{fig:cornerMUBs+outerSICs}).
\end{proof}


\begin{figure}
\centering
\includegraphics[width=0.40\linewidth]{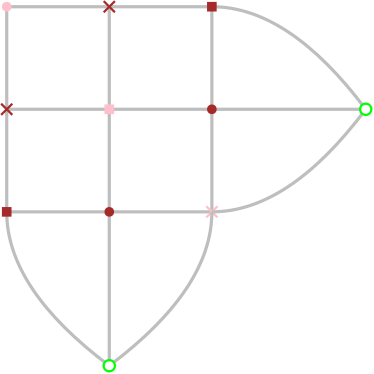}
\caption{Orthogonalities between any of the non-Hesse SIC elements and the corresponding four MUBs. Each SIC element is represented by a gray curve, and each MUB element is represented by one of four markings (circle ---in green---, and disk, square, and cross ---in two colors---) corresponding to its basis. The colors correspond to those of Fig.~\ref{fig:CT93}. A SIC element and a MUB element are orthogonal if, and only if, the marking of the MUB element is on the curve representing the SIC element. The missing (orange) circle corresponds to an element of the Hesse MUB that is not orthogonal to any of the non-Hesse SIC elements.}
\label{fig:cornerMUBs+outerSICs}
\end{figure}


\begin{lemma} \label{lem:2blue}
Consider the Hesse SIC and its MUBs. Any KS assignment with two $1$'s assigned to Hesse SIC elements must assign a $1$ to either one orange Hesse MUB element, or at most one blue Hesse SIC element.
\end{lemma}

\begin{proof}
In any KS assignment of BBC-21 (i.e., the Hesse SIC and Hesse MUBs), at most two of the Hesse SIC vertices can be assigned with $1$.
Neither of the green MUB elements are adjacent to the red Hesse SIC element (which is only adjacent to orange Hesse MUB elements). Therefore, in choosing all green Hesse MUB elements one can always choose the red Hesse SIC element to be a $1$, so that only one blue Hesse SIC element can be assigned $1$. 
\end{proof}

{\em Theorem 4.} The KS-81 is a KS set.

\begin{proof}
KS-81 has $40$ orthogonal bases.
Any KS assignment of $0$'s and $1$'s has to assign a $1$ to one element in each of the $12$ non-Hesse MUBs. Still, there are $28$ more orthogonal bases. The non-Hesse SIC elements can be used to assign a $1$ to, at most, $2 \times 4= 8$ of them.
This leaves, at least, $20$ orthogonal bases to be assigned.
In both cases of Lemma \ref{lem:2blue}, we find that, at least, one orthogonal basis remains unassigned: in the first case, the Hesse MUB elements can assign, at most, $1 + 3 \times 4 = 13$ (one orange, three green) orthogonal bases, and the Hesse SIC elements can assign, at most, $2 \times 3= 6$ (two blue) orthogonal bases, leaving, at least, one orthogonal basis unassigned. In the second case, the Hesse MUB elements can assign, at most, $4 \times 4 = 16$ (four green) orthogonal bases, and the Hesse SIC elements (one blue, one red) can assign, at most, $3$ orthogonal bases. Again, we find that, at least, one orthogonal basis remains unassigned.
\end{proof}


\section{Appendix F: The history of Peres-33, CK-37, CK-33, and CK-31, as told by Peres, and the reason why Conway and Kochen did not use CK-31 in their free-will theorem}


The history of Peres-33, CK-37, CK-33, and CK-31, as told by Asher Peres \cite{Cabello:1996}[pp. 54--55], can be summarized as follows. Around 1990, Peres learned that Simon Kochen and John Conway had found a KS set with 33 elements, CK-33. Roger Penrose gave Peres the idea to visualize CK-33 as follows: ``[The] construction starts from a unit lattice of points in $3$ dimensions. Draw a sphere of radius $2.5$, and keep only the points inside that sphere. Connect them to the center of the sphere. This gives $37$ rays (\ldots). Then remove $4$ ``equatorial'' points'' \cite{Cabello:1996}.
The $33$ remaining points form CK-33. Then, Peres had the idea that ``if, in the cubic lattice, a coordinate $2$ was replaced by $\sqrt{2}$, there would again be numerous orthogonality relations'' \cite{Cabello:1996}. This gives Peres-33. 
However, when Peres wrote Kochen to tell him about Peres-33, Kochen answered that meanwhile Conway and he had found a set with only $31$ elements, CK-31.

Curiously, in their free-will theorem \cite{CK06,CK09,conway_kochen_2011}, Conway and Kochen use Peres-33 rather than CK-31. They explain this choice because Peres-33 is ``more symmetric'' and ``allows a simpler proof'' \cite{CK06}.

\section{Appendix G: Smallest critical KS set inside KS-81}
\label{appG}


The smallest critical KS set inside KS-81 that we have found is the $55$-element set obtained by removing the following $26$ vectors: $v_{1}$, $v_{19}$, $v_{49}$, $v_{50}$, $v_{51}$, $v_{52}$, $v_{53}$, $v_{54}$, $v_{55}$, $v_{56}$, $v_{57}$, $v_{58}$, $v_{59}$, $v_{60}$, $v_{67}$, $v_{68}$, $v_{69}$, $v_{70}$, $v_{71}$, $v_{72}$, $v_{76}$, $v_{77}$, $v_{78}$, $v_{79}$, $v_{80}$, and $v_{81}$. Interestingly, this $55$-element KS set does not contain either the minimum SI-C set or BBC-21. We leave open the question of finding the smallest SI-C set.

The proof that this $55$-element critical KS set is not rigid is as follows: We observe that $19$ of its elements are contained in BBC-21. We compute the general form of any orthogonal representation of the orthogonality graph $G_{19}$ of the $19$-element subset of BBC-21. In particular, we prove that there are two (two-parameter) families which satisfy this (up to unitary transformations). Next, we show that each of these subsets extend to form orthogonal representations of $G_{55}$, the orthogonality graph of the critical $55$-element KS set. In other words, we compute all possible KS sets (up to unitary transformations) whose orthogonalities are the same as the $55$-element KS set. Finally, we provide two orthogonal representations (KS sets) of $G_{55}$ for which there is no unitary mapping one to the other (up to phase). We therefore conclude that the set is not rigid.

\begin{definition}
Given a graph $G = (V,E)$, an assignment $\alpha : V \to \mathbb{C}^d$ is called a \emph{faithful orthogonal representation} of $G$ if vectors $\alpha(u)$ and $\alpha(v)$ are orthogonal if, and only if, $u$ and $v$ are adjacent.
\end{definition}

For brevity, we refer to a faithful orthogonal representation simply as an \emph{orthogonal representation}.

In this appendix and in the following, instead of dealing with KS sets, we often deal with objects that are more restrictive: sequences of vectors in $\mathbb{C}^d$. 

\begin{definition}
 We say that two sequences of vectors $u_1, \ldots, u_k$ and $v_1, \ldots, v_k$ in $\mathbb{C}^d$ are \emph{phase-unitary equivalent} if there exists a unitary $T : \mathbb{C}^d \to \mathbb{C}^d$ and angles $0 \leq \theta_1, \ldots, \theta_k < 2\pi$ satisfying
 \begin{equation}
 v_j = e^{i\theta_j} T(u_j),
 \end{equation}
 for each $j=1, \ldots, k$.
\end{definition}

Clearly, if two KS sets $K$ and $K'$ are equivalent up to some unitary $T$, then there is some ordering of the elements of $K$ and those of $K'$ such that their normalized versions are phase-unitary equivalent. 

\begin{lemma}
 Let $u_1, \ldots, u_{d-1}$ and $v_1, \ldots, v_{d-1}$ be sequences of phase-unitary equivalent vectors in $\mathbb{C}^d$ that are both linearly independent. Let $u_d, v_d \in \mathbb{C}^d$ satisfying $||u_d|| = ||v_d||$, and $<u_j, u_d> = 0$, $<v_j, v_d> = 0$ for each $j=1, \ldots, d-1$. Then, the sequences $u_1, \ldots, u_{d}$ and $v_1, \ldots, v_d$ are also phase-unitary equivalent.
\end{lemma}

\begin{proof}
 Since $T$ is a unitary it preserves inner products and norms. Therefore the vector $T(u_d)$ satisfies each of the conditions of $v_d$, and so, in particular, $||T(u_d)|| = ||v_d||$, and $T(u_d)$ and $v_d$ lie in the same one-dimensional vector space. Therefore, it follows that $v_d = e^{i\theta_d}v_d$ for some $0 \leq \theta_d < 2\pi$.
\end{proof}

The method we use to compute the orthogonal representations of $G_{19}$ relies on \emph{$r$-neighbor bootstrap percolation} \cite{Chalupa:1979, Balogh:1998, Baxter:2010, Balogh:2007, Balogh:2012, Morris:2017, Reichman:2012, Wesolek:2019, Gunderson:2020}. For a vertex $v$ of a graph $G$, the set $N(v)$, called the \emph{neighborhood} of $v$, is the set of vertices of $G$ that are adjacent to $v$.

Let $G = (V,E)$ be a finite graph, let $A_0 \subset V$, and let $r$ be a positive integer. For $i \geq 1$, define $A_i = \{v \in V \setminus A_{i-1} : |N(v) \cap A_{i-1}| \geq r\}$. Since $G$ is finite, this sequence stabilizes (there is some $k$ such that $A_k = A_{k+\ell}$ for any non-negative integer $\ell$). 

We view this as a process (called $r$-neighbor bootstrap percolation), beginning with the set $A_0$, proceeding in rounds from $j=1$ to $k$, at each round generating $A_j$ from $A_{j-1}$. During any round $j$ we call the vertices of $A_j$ \emph{infected}. In each round new vertices become infected when they are adjacent to, at least, $r$ infected vertices (this is the process by which $A_{j+1}$ is obtained from $A_j$). If the process ends with each vertex of the graph infected (i.e., $A_k = V$), then we say that $A_0$ \emph{$r$-percolates} $G$. 

\begin{proposition} \label{prop:percolate}
 Let $K$ be a KS set in ${\cal H}=\mathbb{C}^3$, and let $G = (V,E)$ be the orthogonality graph of $K$. Let $A_0 \subseteq V$ $2$-percolate $G$. Then, the set $S$ of elements of $K$ corresponding to the vertices of $A_0$ fully define the KS set $K$.
\end{proposition}

\begin{proof}
 Let $v \in V$. We show that the element corresponding to $v$ is unique up to multiplication by a complex constant. Since $A_0$ percolates, $v \in A_j$ for some $1 \leq j \leq k$. We proceed by induction on $j$. If $j=0$, $v \in S$ and so is uniquely defined. If $j \geq 1$, then $v$ is orthogonal to two vectors $w_1, w_2 \in A_{j-1}$ that are uniquely defined by $S$. Furthermore, $w_1$ and $w_2$ are linearly independent, and so $v$ lies in the one-dimensional subspace defined by $\langle v,w_1 \rangle = 0$, $\langle v,w_2\rangle = 0$. Therefore, $v$ is defined up to a complex constant.
\end{proof}

Notice that the previous proposition yields a method to explicitly compute an orthogonal representation given an assignment of vectors to the vertices of $A_0$.
For each vertex $v \in V$ that is infected in round $j$, we form a matrix $M_v$ whose rows are the conjugates of the vectors corresponding to the infected neighbors of $v$. Since $v$ already has at least two infected neighbors whose corresponding vectors must be linearly independent we see that $\rank(M_v) \in \{2,3\}$. If $\rank(M_v) = 3$, then the assignment of vectors to the vertices of $A_0$ does \emph{not} extend to a KS set. Otherwise, the vector corresponding to $v$ is any generator of its one dimensional null space of $M_v$ (and is thus defined up to a complex constant). 

By choosing symbolic vectors (i.e., vectors whose components have $\mathbb{C}^d$-valued variables) for the $2$-percolating set we can, in principle, compute all possible orthogonal representations (up to unitary transformations) by percolating and solving equations that ensure that each $M_v$ has rank $2$. This process may require several ``rounds'' in which we eliminate variables using the obtained equations, redefine the initial vectors, and apply the percolation process again.

The 19-element subset is:
\begin{equation}
    \begin{array}{llllll}
        v_1 = (-1,0,1), & v_2 = (1,-1,0), & v_3 = (-\omega, 0, 1), & v_4 = (1, -\omega^2, 0), & v_5 = (-\omega^2,0,1), & v_6 = (1, -\omega, 0), \\
        v_7  = (0, 1, -\omega^2), & v_8 = (0,1,-\omega), & v_9 = (1,1,1), & v_{10} = (1, \omega, \omega^2), & v_{11} = (1, \omega^2, \omega), & v_{12} = (1,\omega^2,\omega^2), \\
        v_{13} = (\omega^2, \omega^2, 1), & v_{14} = (\omega^2,1,\omega^2), & v_{15} = (1, \omega, \omega), & v_{16} = (\omega, \omega, 1), & v_{17} = (\omega, 1, \omega), & v_{18} = (0,1,0), \\ v_{19} = (0,0,1).
    \end{array}
\end{equation}
Denote the vertices of its orthogonality graph $G_{19}$ by $z_1,\ldots,z_{19}$ so that $z_i$ and $z_j$ are adjacent if, and only if, the vectors $v_i$ and $v_j$ are orthogonal. 
Vertices $z_1$, $z_2$, $z_3$, $z_9$, $z_{10}$, $z_{11}$ percolate $G_{19}$, with vertices $z_9,z_{10},z_{11}$ forming a clique of size $3$. As we are interested in orthogonal representations up to the application of a unitary (and up to phase), we may set $w_9 = (1,0,0), w_{10} = (0,1,0), w_{11} = (0,0,1)$. Moreover, $z_1$ and $z_2$ are both adjacent to $z_9$, and $z_3$ is adjacent to $z_{10}$ so we assign them $w_1 = (0,1,a), w_2 = (0,1,b)$, and $w_3 = (1,0,c)$, respectively, for $a,b,c \in \mathbb{C}$.
When we percolate with these vectors, we find that everything is defined with the exception of the vectors corresponding to $z_4,z_{12},z_{15}$. Table~\ref{tab:19ray-partialpercolation} illustrates this (partial) percolation process. 

We obtain equations associated to each of these vertices $v \in \{z_{4}, z_{12}, z_{15}\}$ in terms of $a$, $b$, and $c$, since the rank of $M_v$ must be $2$, and so the determinant of every $3 \times 3$ sub-matrix of $M_v$ must be $0$. To simplify matters, we only display the numerator of the determinant for each of the equations, and label each of the equations by the orthogonalities that generate it.


\begin{table}
	\centering
	\begin{tabular}{cccc}
            \hline \hline
		Round & Vertex & Infected neighbors & Vector $w_j$  \\
		\hline
		\multirow{6}{*}{0} & 9 & ~ & $\left(1,\,0,\,0\right)$ \\
		& $10$ & ~ & $\left(0,\,1,\,0\right)$ \\
		& $11$ & ~ & $\left(0,\,0,\,1\right)$ \\
		& $1$  & ~ & $\left(0,\,1,\,a\right)$ \\
		& $2$  & ~ & $\left(0,\,1,\,b\right)$ \\
		& $3$  & ~ & $\left(1,\,0,\,c\right)$ \\
		\hline
		\multirow{2}{*}{1} 
		& $16$ & $2,3$  & $\left(1,\,\frac{b^*}{c^*},\,-\frac{1}{c^*}\right)$ \\
		& $18$ & $1,3$ & $\left(1,\,\frac{a^*}{c^*},\,-\frac{1}{c^*}\right)$ \\
		\hline
		\multirow{4}{*}{2} 
		& $5$  & $11,18$ & $\left(1,\,-\frac{c}{a},\,0\right)$ \\
		& $7$  & $11,16$ & $\left(1,\,-\frac{c}{b},\,0\right)$ \\
		& $17$ & $1,16$  & $\left(1,\,-\frac{c a^*}{b a^* + 1},\,\frac{c}{b a^* + 1}\right)$ \\
		& $19$ & $2,18$  & $\left(1,\,-\frac{c b^*}{a b^* + 1},\,\frac{c}{a b^* + 1}\right)$ \\
		\hline
		\multirow{6}{*}{3} 
		& $4$  & $11,17,19$ & ~ \\
		& $6$  & $10,19$    & $\left(1,\,0,\,-\frac{b a^* + 1}{c^*}\right)$ \\
		& $8$  & $10,17$    & $\left(1,\,0,\,-\frac{a b^* + 1}{c^*}\right)$ \\
		& $13$ & $2,5$      & $\left(1,\,\frac{a^*}{c^*},\,-\frac{a^*}{b^* c^*}\right)$ \\
		& $14$ & $1,7$      & $\left(1,\,\frac{b^*}{c^*},\,-\frac{b^*}{a^* c^*}\right)$ \\
		& $15$ & $5,16,17$  & ~ \\
		\hline
		\multirow{3}{*}{4} 
		& $4$  & $11,17,19$ & ~ \\
		& $12$ & $3,13,14$  & ~ \\
		& $15$ & $5,6,16,17$ & ~ \\
		\hline \hline
	\end{tabular}
	\caption{2-neighbor bootstrap percolation process on the orthogonality graph of the $19$-element set with $A_0 = \{1,2,3,9,10,11\}$ (indicated by Round 0). The vertices $z_4,z_{15}$ appear twice in the table (in Rounds 3 and 4). Since the corresponding vectors $w_4,w_{15}$ remain unassigned, we do not add them to the infected list. In Round 4, the vertex $z_{15}$ has an additional infected neighbor ($z_6$, which was infected in Round 3), and therefore, we obtain additional equations.}
	\label{tab:19ray-partialpercolation}
\end{table}


The vertex $z_{12}$ is adjacent to previously infected vertices $z_3,z_{13},z_{14}$. From these, we obtain
\begin{align}
    (a-b)(|c|^2-1) & = 0 \ & (z_{12} \sim z_{3},z_{13},z_{14}),
\end{align}
where $z_i \sim z_j$ denotes that $z_i$ and $z_j$ are adjacent in the graph. Since $a \neq b$, we find that $|c| = 1$. Therefore, we may set $c = e^{i\gamma}$.

We also obtain equations from $z_{4}$ and $z_{15}$ (which we have simplified with $|c|=1$): 
\begin{align}
    |a|^2b - a|b|^2 + a - b &= 0 & (z_4 \sim z_{11}, z_{17}, z_{19}), \label{eq:all-solve}\\
    |a|^2|b|^2 + a^*b + ab^* &= 0 & (z_{15} \sim z_5,z_6,z_{16}), \label{eq:norm-solve} \\
    |a|^2 - a^*b - ab^* - 2 &= 0 & (z_{15} \sim z_5,z_{16}, z_{17}), \label{eq:conj-solve}
\end{align}
and from Eqs.~\eqref{eq:norm-solve} and \eqref{eq:conj-solve}, we find that
\begin{equation} \label{eq:norms}
    |a|^2(|b|^2+1) = 2.
\end{equation}
Rearranging Eq.~\eqref{eq:all-solve}, we have $a(1-|b|^2) = b(1-|a|^2)$, from which it follows that either (i) $|b|^2 = 1$ or (ii) $b/a = |b|^2 + 1$ [by solving for $|a|^2$ in Eq.~\eqref{eq:norms}]. By plugging $b = a(|b|^2+1)$ into Eq.~\eqref{eq:all-solve} we may show that in case (ii) it also follows that $|b|^2 = 1$. By Eq.~\eqref{eq:norms}, we have that $|a|^2 = 1$, and so $|a| = |b|=1$.

Therefore, $a = e^{i\alpha}, b = e^{i\beta}$, and so
the real part of Eq.~\eqref{eq:norm-solve} may be written as
\begin{equation}
    \cos(\alpha-\beta) = -\frac{1}{2}.
\end{equation}
We therefore find that $\alpha -\beta = \pm 2\pi/3 + 2\pi n$ must be satisfied for some $n \in \mathbb{Z}$, and so, in particular, that
\begin{equation}
    b \in \{\omega a, \omega^2 a\},
\end{equation}
where $\omega = e^{2i\pi/3}$.
Both of these choices generate orthogonal representations of $G_{19}$, which are both illustrated in Table \ref{tab:19ray-generalform}. Moreover, we have checked that both of these orthogonal representations do extend to orthogonal representations of $G_{55}$ (by percolating the $19$ elements).


\begin{table}[]
    \centering
    \begin{tabular}{cll}
    \hline \hline
        Index $i$ & Vector $w_i$ & Vector $w'_i$ \\
        \hline 
        $1$ & $\left(0,\,1,\,e^{i\alpha}\right)$ & $\left(0,1,e^{i\alpha}\right)$ \\ 
        $2$ & $\left(0,\,1,\,e^{i\left(\alpha + 2\pi/3\right)}\right)$ & $\left(0,1,e^{i(\alpha-2\pi/3)}\right)$ \\ 
        $3$ & $\left(1,\,0,\,e^{i\gamma}\right)$ & $\left(1,0,e^{i\gamma}\right)$ \\ 
        $4$ & $\left(1,e^{-i\left( \, \alpha -\gamma + \pi/3\right)},\,0\right)$ & $\left(1,e^{-i(\alpha-\gamma - \pi/3)},0\right)$\\ 
        $5$ & $\left(1,\,-e^{-i\left( \, \alpha- \gamma\right)},\,0\right)$ & $\left(1,-e^{-i(\alpha-\gamma)},0\right)$\\
        $6$ & $\left(1,\,0, e^{i(\gamma-2\pi/3)}\right)$ & $\left(1,0,e^{i(\gamma + 2\pi/3)}\right)$ \\ 
        $7$ & $\left(1,e^{-i(\alpha-\gamma-\pi/3)},\,0\right)$ & $\left(1,e^{-i(\alpha-\gamma+\pi/3)},0\right)$\\ 
        $8$ & $\left(1,\,0, e^{i(\gamma+2\pi/3)}\right)$ & $\left(1,0,e^{i(\gamma -2\pi/3)}\right)$ \\ 
        $9$ & $\left(1,\,0,\,0\right)$ & $(1,0,0)$\\ 
        $10$ & $\left(0,\,1,\,0\right)$ & $(0,1,0)$ \\ 
        $11$ & $\left(0,\,0,\,1\right)$ & $(0,0,1)$ \\ 
        $12$ & $\left(1,e^{-i(\alpha-\gamma-2\pi/3)},\,-e^{i\gamma}\right)$ & $\left(1,e^{-i(\alpha-\gamma+2\pi/3)},e^{-i\gamma}\right)$ \\ 
        $13$ & $\left(1,\,e^{-i(\alpha-\gamma)}, e^{i(\gamma-\pi/3)}\right)$ & $\left(1, e^{-i(\alpha-\gamma)}, e^{i(\gamma+\pi/3)}\right)$ \\ 
        $14$ & $\left(1,e^{-i(\alpha-\gamma+2\pi/3)}, e^{i(\gamma+\pi/3)}\right)$ & $\left(1,e^{-i(\alpha-\gamma-2\pi/3)}, e^{i(\gamma-\pi/3)}\right)$\\ 
        $15$ & $\left(1,\,e^{-i(\alpha-\gamma)},e^{i(\gamma+\pi/3)}\right)$ & $\left(1, e^{-i(\alpha-\gamma)}, e^{i(\gamma-\pi/3)}\right)$ \\ 
        $16$ & $\left(1, e^{-i(\alpha-\gamma+2\pi/3)},\,-e^{i\gamma}\right)$ & $\left(1,e^{-i(\alpha-\gamma-2\pi/3)}, e^{-i\gamma}\right)$ \\ 
        $17$ & $\left(1,\,e^{-i(\alpha-\gamma-2\pi/3)},e^{i(\gamma-\pi/3)}\right)$ & $\left(1, e^{-i(\alpha-\gamma+2\pi/3)}, e^{i(\gamma+\pi/3)}\right)$ \\ 
        $18$ & $\left(1,\,e^{-i(\alpha-\gamma)},\,-e^{i\gamma}\right)$ & $\left(1,e^{-i(\alpha-\gamma)}, e^{-i\gamma}\right)$\\ 
        $19$ & $\left(1,e^{-i(\alpha-\gamma-2\pi/3)},\,e^{i(\gamma+\pi/3)}\right)$ & $\left(1,e^{-i(\alpha-\gamma+2\pi/3)}, e^{i(\gamma-\pi/3)}\right)$ \\ 
    \hline \hline
    \end{tabular}
    \caption{The set of elements $w_i$ obtained by percolating with $b=\omega a$, and the set of elements $w'_i$ obtained by percolating with $b=-\omega a$. The $w'_i$ may be obtained from the $w_i$ by applying $\pi/3 \to -\pi/3$ (and thus $2\pi/3 \to -2\pi/3$). }
    \label{tab:19ray-generalform}
\end{table}


Now let us consider the sets of $55$ elements $\mathcal{N}$ and $\mathcal{N}'$ obtained by setting $\alpha=0, \beta=2\pi/3, \gamma = 0$, and $\alpha=0, \beta=-2\pi/3,\gamma = 0$, respectively. We prove that there is no unitary mapping $\mathcal{N}$ to $\mathcal{N}'$ (up to phase). 

Let us choose some ordering of the $55$ elements so that the first $19$ are in the same order as for the $19$-element set (the order of the remaining elements is not essential). Call $z_1, \ldots, z_{55}$ the vertices of $G_{55}$ in the ordering just defined (so that the graph induced on the first $19$ vertices is $G_{19}$ and the vertices appear in the same order as in Table~\ref{tab:19ray-generalform}). Let us also apply this ordering to obtain sequences of unit vectors $(u_i), (u'_i)$ from $\mathcal{N}$ and $\mathcal{N}'$, respectively (so that $u_1,\ldots,u_{19}$ are obtained by plugging in $\alpha=0,\beta=2\pi/3,\gamma=0$ into the middle column of Table~\ref{tab:19ray-generalform}, and then normalizing. Similar for the $u'_i$ by plugging in $\alpha=0$,$\beta=-2\pi/3$, and $\gamma=0$ into the last column).

The automorphism group of $G_{55}$ has size two. The only non-trivial automorphism $\sigma$ acts on the percolating set by interchanging the pairs $(z_1,z_2)$, $(z_3,z_4)$, $(z_{10},z_{11})$, and fixing $z_9$. 
Since unitaries preserve inner products, they preserve the orthogonality graph. Therefore, if there is a unitary $T : \mathbb{C}^3 \to \mathbb{C}^3$ mapping $\mathcal{N}$ to $\mathcal{N}'$, then it must be that either (a) $e^{i\theta_j}T(u_j) = u'_j$ for each $j=1,\ldots,55$, or (b) $e^{i\theta_j}T(u_j) = u'_{\sigma(j)}$ for each $j=1,\ldots,55$. We prove now that both of these cases are impossible.

In Case (a), we have that $e^{i\theta_9}T(u_9) = u'_9$, $e^{i\theta_{10}}T(u_{10}) = u'_{10}$, and $e^{i\theta_{11}}T(u_{11}) = u'_{11}$ (i.e., the unitary maps the standard basis to itself). Therefore, the matrix representing the unitary $T$ is of the form
\begin{equation}
    \begin{pmatrix}
        e^{i\phi_1} & 0 & 0 \\
        0 & e^{i\phi_2} & 0 \\
        0 & 0 & e^{i\phi_3}
    \end{pmatrix}
\end{equation}
for some $0 \leq \phi_1,\phi_2,\phi_3 < 2\pi$.

In order to simplify notation, in the following we write $x \equiv y$ if $x = y + 2\pi n$ for some $n \in \mathbb{Z}$.
From $e^{i\theta_1}T(u_1) = u'_1$, we obtain that 
\begin{align}
    \theta_1 + \phi_2 &\equiv 0, \label{eq:u1.1} \\
    \theta_1 + \phi_3 &\equiv 0, \label{eq:u1.2}
\end{align}
and from $e^{i\theta_2}T(u_2) = u'_2$, we obtain that
\begin{align}
    \theta_2 + \phi_2 &\equiv 0, \label{eq:u2.1} \\
    \theta_2 + \phi_3 + 2\pi/3 &\equiv - 2\pi/3. \label{eq:u2.2}
\end{align}
Therefore, $-\theta_1 + \theta_2 \equiv 0$ [from Eqs.~\eqref{eq:u1.1} and \eqref{eq:u2.1}] and $-\theta_1 + \theta_2 \equiv -4\pi/3$ [from Eqs.~\eqref{eq:u1.2} and \eqref{eq:u2.2}]. However, this implies that $-4\pi/3 \equiv 0$, which is a contradiction.

In Case (b), the matrix representing the unitary must be of the form
\begin{equation}
    \begin{pmatrix}
        e^{i\phi_1} & 0 & 0 \\
        0 & 0 & e^{i\phi_3} \\
        0 & e^{i\phi_2} & 0
    \end{pmatrix}
\end{equation}
for some angles $0 \leq \phi_1,\phi_2,\phi_3 < 2\pi.$
From $e^{i\theta_1}T(u_1) = u'_2$, we obtain
\begin{align}
    \theta_1 + \phi_3 &\equiv 0 \label{eq:u1.1-auto}, \\
    \theta_1 + \phi_2 &\equiv -2\pi/3 \label{eq:u1.2-auto},
\end{align}
and, from $e^{i\theta_1}T(u_2) = u'_1$, we obtain
\begin{align}
    \theta_2 + \phi_3 + 2\pi/3 &\equiv 0 \label{eq:u2.1-auto},\\
    \theta_2 + \phi_2 &\equiv 0. \label{eq:u2.2-auto}
\end{align}
Therefore, $-\theta_1 + \theta_2 \equiv -2\pi/3$ [from Eqs.~\eqref{eq:u1.1-auto} and \eqref{eq:u2.1-auto}] and $-\theta_1 + \theta_2 \equiv 2\pi/3$ [from Eqs.~\eqref{eq:u1.2-auto} and \eqref{eq:u2.2-auto}]. This time we find that $4\pi/3 \equiv 0$, which is again a contradiction.


\section{Appendix H: An alternate proof of the rigidity of CK-31} \label{app:CK31-alternate}


Here, we apply the same techniques used in Appendix~G, this time to prove that CK-31 is rigid.

Consider the following ordering of CK-31:
\begin{equation}
\begin{array}{llllll}
v_{1}=(1, 0, 0), & v_{2}=(0, 1, 0),&
v_{3}=(0, 0, 1), & v_{4}=(1, 1, 0),&
v_{5}=(-1, 1, 0), & v_{6}=(1, 0, 1),\\
v_{7}=(1, 0, -1), & v_{8}=(0, 1, 1),&
v_{9}=(0, -1, 1), & v_{10}=(-2, 1, 0),&
v_{11}=(1, 2, 0), & v_{12}=(2, 0, 1),\\
v_{13}=(-2, 0, 1), & v_{14}=(1, 1, 1),&
v_{15}=(-1, 1, 1), & v_{16}=(1, -1, 1),&
v_{17}=(1, 1, -1), & v_{18}=(0, 2, 1),\\
v_{19}=(0, -2, 1), & v_{20}=(1, 0, 2),&
v_{21}=(1, 0, -2), & v_{22}=(0, 1, 2),&
v_{23}=(0, 1, -2), & v_{24}=(-2, 1, 1),\\
v_{25}=(2, -1, 1), & v_{26}=(1, 2, 1),&
v_{27}=(1, 2, -1), & v_{28}=(1, 1, 2),&
v_{29}=(-1, 1, 2), & v_{30}=(1, -1, 2),\\
v_{31}=(1, 1, -2), & 
\end{array}
\end{equation}
defining the sequence of vectors $v_1,\ldots,v_{31}.$

\begin{lemma} \label{lem:CK37-percolate}
 Let $G$ be the orthogonality graph of CK-31.
 Vertices $1$,$2$,$3$,$17$ $2$-percolate $G$. 
\end{lemma}

\begin{proof}
 See Table \ref{tab:CK31-percolation-partial}. At each round, the new vertices introduced each have at least two infected neighbors. 
\end{proof}

Note that one of the vertices $1$,$2$,$3$ can also be removed from the initial infected set. We keep each of them for the sake of convenience.


\begin{table} \label{tab:percolation:partial}
 \begin{tabular}{cccc}
 \hline \hline
 Round & Vertex & Infected neighbors & Vector \\
 \hline 
 \multirow{ 4}{*}{$0$} 
 & $1$ & & $(1, 0, 0)$ \\
 & $2$ & & $(0, 1, 0)$\\
 & $3$ & & $(0, 0, 1)$\\
 & $17$ & & $(1, b, c)$ \\
 \hline 
 \multirow{ 3}{*}{$1$} 
 & $5$ & $3,17$ & $(1, -\frac{1}{b^*}, 0)$ \\
 & $6$ & $2,17$ & $(1, 0, -\frac{1}{c^*})$ \\
 & $8$ & $1,17$ & $(0, 1, -\frac{b^*}{c^*})$ \\
 \hline 
 \multirow{ 5}{*}{$2$} 
 & $4$ & $3,5$ & $(1, b, 0)$ \\
 & $7$ & $2,6$ & $(1, 0, c)$ \\
 & $9$ & $1, 8$ & $(0, 1, \frac{c}{b})$ \\
 & $25$ & $8,17$ & $(1, b, c)$ \\
 & $28$ & $5,17$ & \\
 \hline
 \multirow{ 7}{*}{$3$} 
 & $11$ & $3,25$ \\
 & $13$ & $2,28$ \\
 & $14$ & $5,7,9$ \\
 & $15$ & $4,6,9$ \\
 & $16$ & $4,7,8$ \\
 & $19$ & $1,28$ \\ 
 & $21$ & $2,25$ \\
 \hline 
 \multirow{ 10}{*}{$4$} 
 & $10$ & $3,11$ \\
 & $12$ & $2,21$ \\
 & $20$ & $2,13$ \\ 
 & $22$ & $1,19$ \\
 & $24$ & $9,11,14$ \\
 & $26$ & $7,16$ \\
 & $27$ & $6,15$ \\
 & $29$ & $4,16,19$ \\
 & $30$ & $4,13,15$ \\
 & $31$ & $5,14 $\\
 \hline 
 \multirow{ 2}{*}{$5$} 
 & $18$ & $1,30,31$ \\
 & $23$ & $1,26$ \\
 \hline \hline
 \end{tabular}
 \caption{The $2$-neighbor bootstrap percolation process on the orthogonality graph of CK-31 starting with $A_0 = \{1,2,3,17\}$. For each round from $1$ to $5$, we indicate the newly infected vertices. We also indicate the initial infected set by Round 0, and the vectors obtained before restrictions must be placed on $b,c$.}
 \label{tab:CK31-percolation-partial}
\end{table}


We now prove that CK-31 is rigid. Denote its orthogonality graph by $G$, and its vertices by $z_1, \ldots, z_{31}$. Let us consider some other KS $K' = \{w_1, \ldots, w_{31}\}$ that also has orthogonality graph $G$ (so $w_j$ and $w_k$ are orthogonal if, and only if, $z_j$ and $z_k$ are adjacent vertices in $G$).
Since vertices $z_1, z_2, z_3$ are pairwise adjacent, they must correspond to an orthogonal basis in $\mathbb{C}^d$. Therefore, we may assume that they are the standard basis vectors $(1,0,0)$, $(0,1,0)$, $(0,0,1)$, respectively (by just applying the appropriate unitary). Let us denote by $(a,b,c)$ the vector corresponding to vertex $z_{17}$. Since $z_{17}$ is not adjacent to $z_1$, $a \neq 0$, and so we can assume that $a=1$. One can then determine the vectors (in terms of the variables $b,c$) corresponding to the vertices obtained in the first two rounds of the percolation process ($z_5,z_6,z_8$ for round 1 and $z_4,z_7,z_9,z_{25},z_{28}$ for round 2).

Vertex $z_{14}$ is orthogonal to each of $z_5,z_7,z_9$, and so its vector is orthogonal to each. Therefore, we find that the matrix whose rows are $w_{5}^*, w_{7}^*, w_{9}^*$ must have determinant $0$. That is,
\begin{align}
 \begin{vmatrix}
 1 & -\frac{1}{b} & 0 \\
 1 & 0 & c^* \\
 0 & 1 & \frac{c^*}{b^*} \\
 \end{vmatrix} = 0, 
\end{align}
and so we obtain the equation
\begin{equation} 
 -c^*\left(\frac{1}{|b|^2} - 1\right) = 0.
\end{equation}
Therefore, we either have that $c = 0$ or that $|b| = 1$. The solution $c = 0 $ may be eliminated by the fact that $v_{17}$ is not orthogonal to $v_3$, and so we find that $|b| = 1$.

Similarly, from $z_{16}$, we find that the matrix whose rows are $w_{4}^*, w_{7}^*, w_{8}^*$ has determinant $0$. That is,
\begin{align}
 \begin{vmatrix}
 1 & b^* & 0 \\
 1 & 0 & c^* \\
 0 & 1 & -\frac{b}{c}
\end{vmatrix} = 0, 
\end{align}
and so we obtain the equation
\begin{equation}
 |b|^2 = |c|^2.
\end{equation}
Since we concluded earlier that $|b| = 1$, it follows that $|c| = 1 $ as well. Therefore, we may denote the vector corresponding to vertex 13 by $(1, e^{i\phi}, e^{i\gamma})$ for some angles $0 \leq \phi, \gamma < 2\pi$. In Table~\ref{tab:generators}, we show the vectors corresponding to the $2$-percolating set $1,2,3,17$ for CK-31 and for $K'$. In Table~\ref{tab:CK31-percolation}, we illustrate the entire percolation process and the resulting KS set.


\begin{table}[]
 \centering
 \begin{tabular}{cccc}
 \hline \hline
 $j$ & $v_j$ & $w_j$ & $\theta_j$\\
 \hline 
 $1$ & $(1,0,0)$ & $(1, 0, 0)$ & $0$\\
 $2$ & $(0,1,0)$ & $(0, 1, 0)$ & $-\phi$ \\
 $3$ & $(0,0,1)$ & $(0, 0, 1)$ & $-\gamma + \pi$ \\
 $17$ & $(1,1,-1)$ & $(1, e^{i\phi}, e^{i\gamma})$ & $0$ \\
 \hline \hline
 \end{tabular}
 \caption{The sequence of vectors $v_j$ determining CK-31 and the set of vectors $w_j$ determining $K'$ are phase-unitary equivalent. For each $j \in \{1,2,3,17\}$, $w_j = e^{i\theta_j}T(v_j)$, where $T$ is the unitary defined in Eq.~\eqref{eq:the-unitary}.}
 \label{tab:generators}
\end{table}


\begin{lemma}
 The sequences $v_{1}, v_2, v_3, v_{17}$ and $ w_{1}, w_2, w_3, w_{17}$ are phase-unitary equivalent.
\end{lemma}

\begin{proof}
 Define the unitary $T$ via the matrix
 \begin{equation} \label{eq:the-unitary}
 T = 
 \begin{pmatrix}
 1 & 0 & 0 \\
 0 & e^{i\phi} & 0 \\
 0 & 0 & -e^{i\gamma}
 \end{pmatrix},
 \end{equation}
 and angles $\theta_{1} := 0$, $\theta_{2} := -\phi + \pi$, and $\theta_3 := -\gamma, \theta_{17} = 0$. Then, $w_j = e^{i\theta_j}T(v_j)$ for each $j \in \{1,2,3,17\}$.
\end{proof}

By Lemmas \ref{prop:percolate} and \ref{lem:CK37-percolate}, it follows that the normalized version of CK-31 (i.e., the sequence of normalized vectors $v_1/|v_1|,\ldots,v_{31}/|v_{31}|$) and the normalized version of $K'$ are phase-unitary equivalent. Thus we have proven our main result. \\

{\em Theorem 1.} CK-31 is rigid. \\

The technique described in this section and the previous is general. For example, we have used it to confirm the rigidity of the minimal SI-C set, and also to confirm the non-rigidity of Peres-33 and Penrose-33 \cite{gould2010isomorphism, bengtsson2012gleason,Xu:2024PRL}.
Moreover, one can generate KS sets from the orthogonality graph (i.e., compute orthogonal representations) in this manner --- choosing a small percolating set, assigning vectors, and percolating. This process yields not only some KS set with these orthogonalities, but the general form of \emph{any} KS set satisfying the orthogonalities. This may prove to be useful practically since percolating sets can be significantly smaller than the KS sets they percolate (in the case of CK-31 we used only 4 of the 31 vectors in the percolating set).


\begin{table*}
 \begin{tabular}{ccccc}
 \hline \hline
 Round & Vertex & Infected neighbors & KS vector & Infected graph \\
 \hline 
 & & & & \multirow{9}{*}{\includegraphics[width=0.25\textwidth]{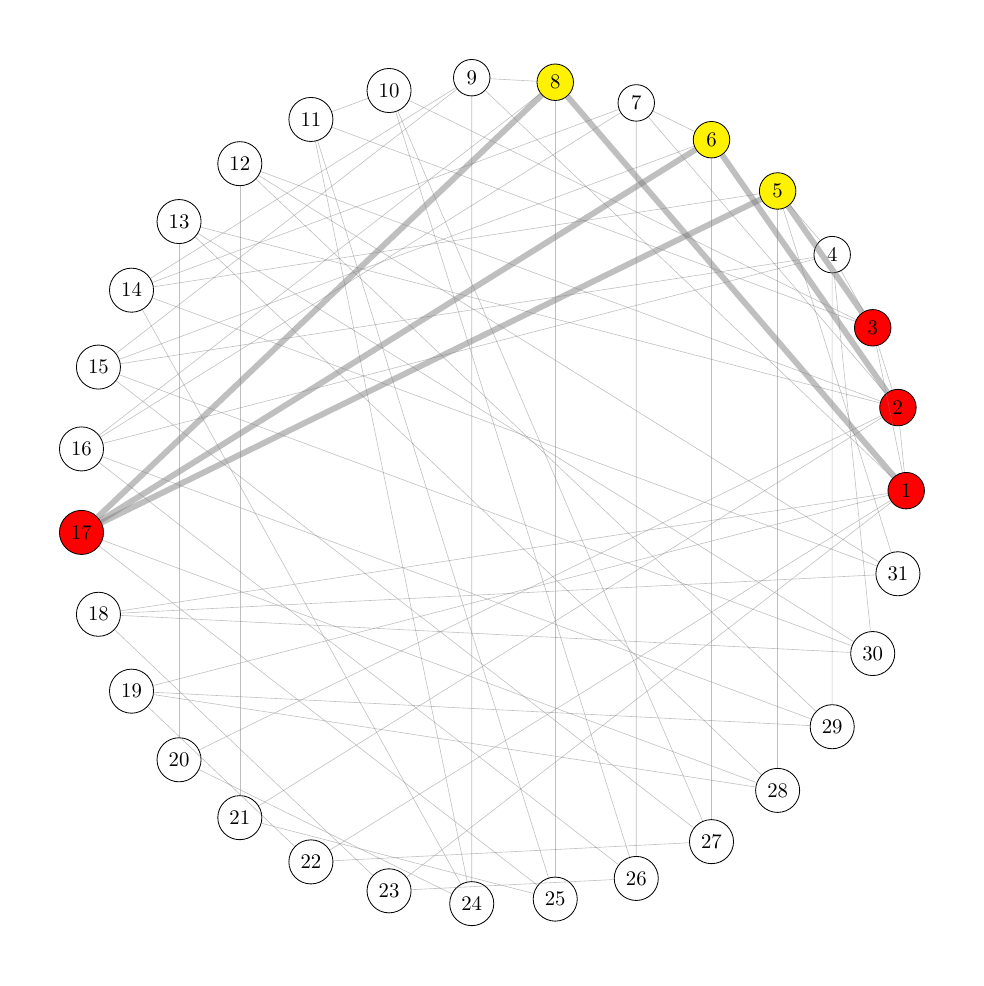}}\\
 \multirow{ 8}{*}{$0$} 
 & $1$ & & $(1, 0, 0)$ & \\
 & & & & \\
 & $2$ & & $(0, 1, 0)$\\
 & & & & \\
 & $3$ & & $(0, 0, 1)$\\
 & & & & \\
 & $17$ & & $(1, e^{i\phi}, e^{i\gamma})$ \\
 & & & & \\
 \hline 
 & & & & \multirow{9}{*}{\includegraphics[width=0.25\textwidth]{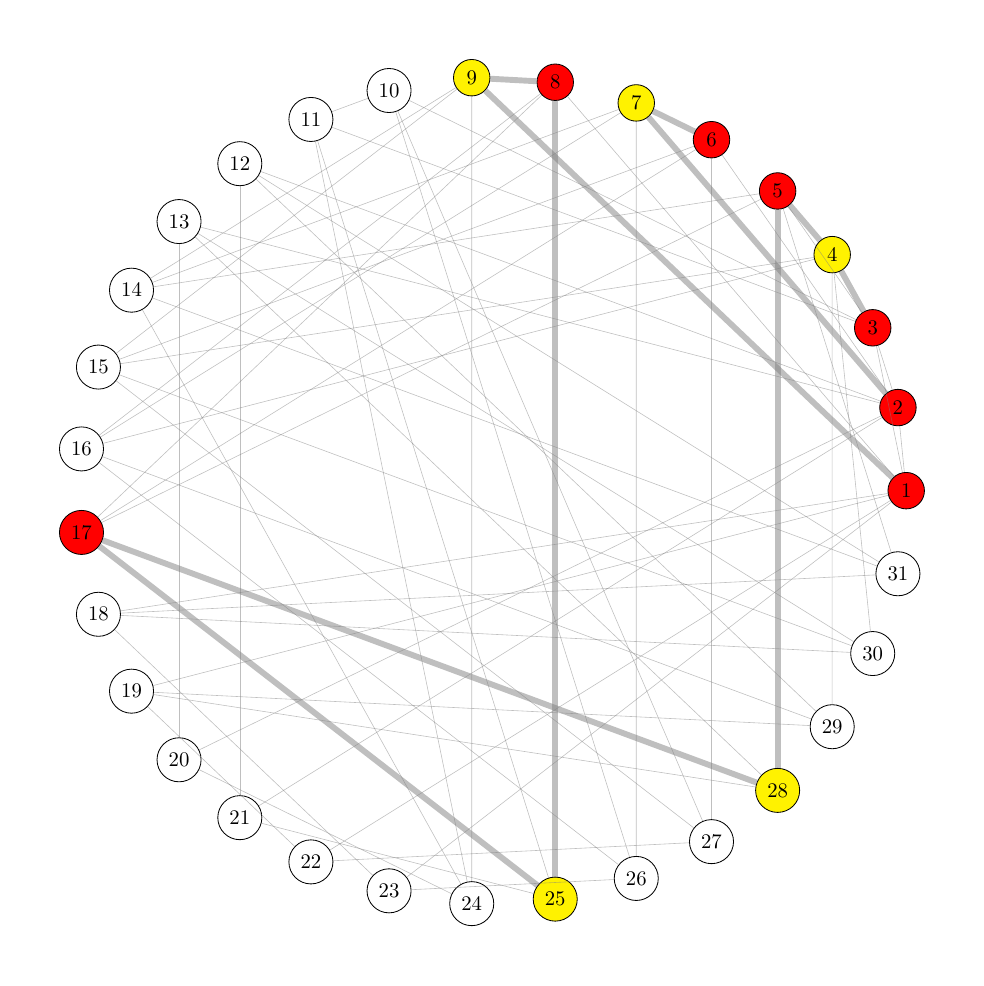}} \\
 \multirow{ 8}{*}{$1$} 
 & & & & \\
 & $5$ & $3,17$ & $(1, -e^{i\phi}, 0)$ \\
 & & & & \\
 & $6$ & $2,17$ & $(1, 0, -e^{i\gamma})$ \\
 & & & & \\
 & $8$ & $1,17$ & $(0, 1, -e^{-i(\phi - \gamma)})$ \\
 & & & & \\
 & & & & \\
 \hline 
 & & & & \multirow{9}{*}{\includegraphics[width=0.25\textwidth]{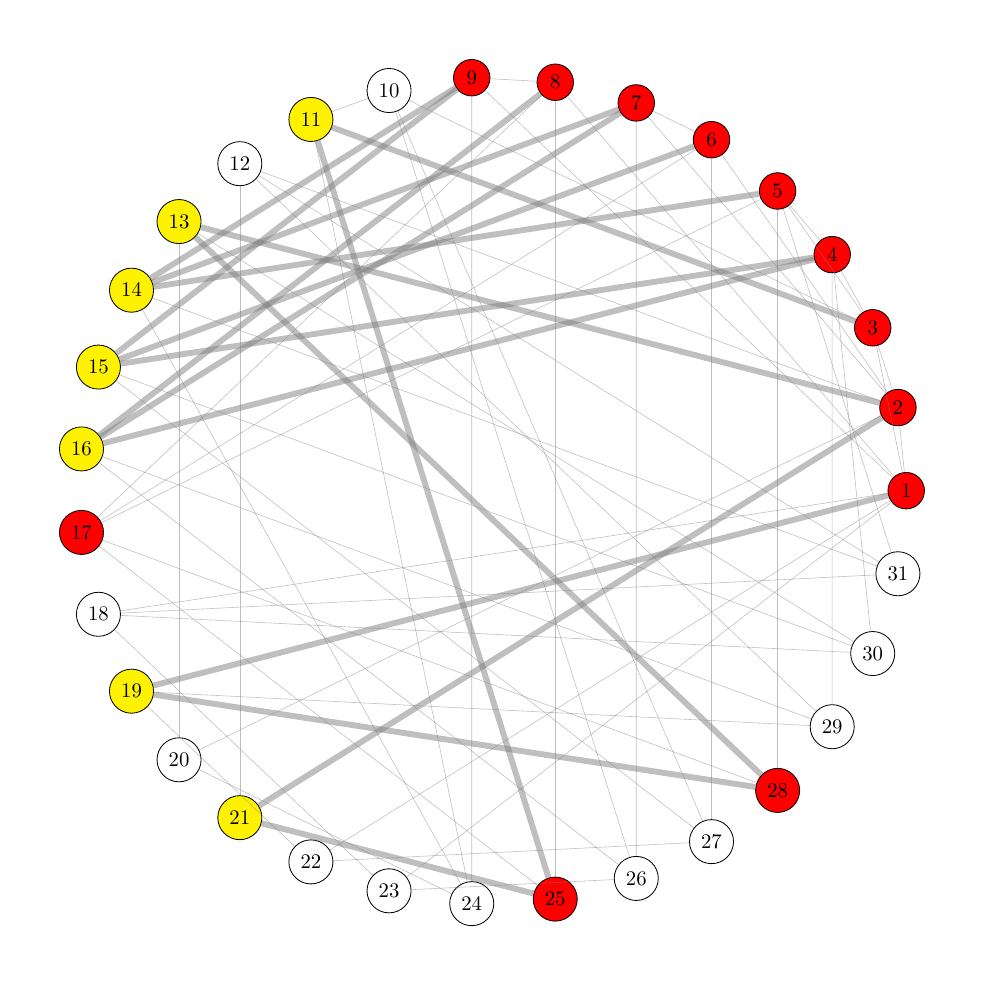}} \\[-5pt]
 \multirow{ 7}{*}{$2$} 
 & $4$ & $3,5$ & $(1, e^{i\phi}, 0)$ & \\[-2pt]
 & & & & \\[-2pt]
 & $7$ & $2,6$ & $(1, 0, e^{i\gamma})$ \\[-3pt]
 & & & & \\[-2pt]
 & $9$ & $1, 8$ & $(0, 1, e^{-i(\phi + \gamma)})$ \\[-2pt]
 & & & & \\[-2pt]
 & $25$ & $8,17$ & $(1, -e^{i\phi}/2, -e^{i\gamma}/2)$ \\[-2pt]
 & & & & \\[-2pt]
 & $28$ & $5,17$ & $(1, e^{i\phi}, -2e^{i\gamma}) $& \\[-2pt]
 & & & & \\[-2pt]
 \hline
 & & & & \multirow{9}{*}{\includegraphics[width=0.25\textwidth]{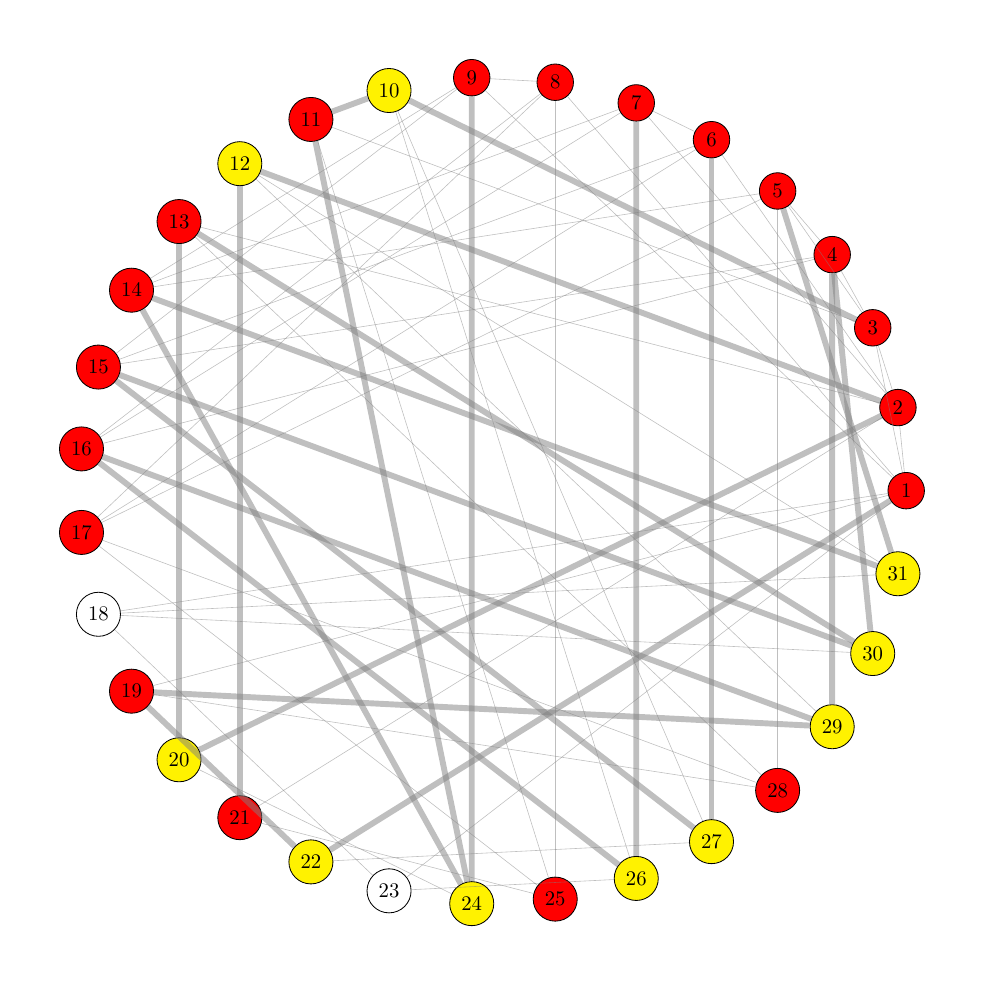}} \\[-5pt]
 \multirow{ 8}{*}{$3$} 
 & $11$ & $3,25$ & $(1, 2e^{i\phi}, 0)$ & \\
 & & & & \\[-12pt]
 & $13$ & $2,28$ & $(1, 0, e^{i\gamma}/2)$ \\
 & & & & \\[-12pt]
 & $14$ & $5,7,9$ & $(1, e^{i\phi}, -e^{i\gamma})$\\
 & & & & \\[-12pt]
 & $15$ & $4,6,9$ & $(1, -e^{i\phi}, e^{i\gamma})$ \\
 & & & & \\[-12pt]
 & $16$ & $4,7,8$ & $(1, -e^{i\phi}, -e^{i\gamma})$ \\
 & & & & \\[-12pt]
 & $19$ & $1,28$ & $(0, 1, e^{-i(\phi - \gamma)}/2)$\\ 
 & & & & \\[-12pt]
 & $21$ & $2,25$ & $(1, 0, 2e^{i\gamma})$ \\
 & & & & \\[-5pt]
 \hline 
 & & & & \multirow{11}{*}{\includegraphics[width=0.25\textwidth]{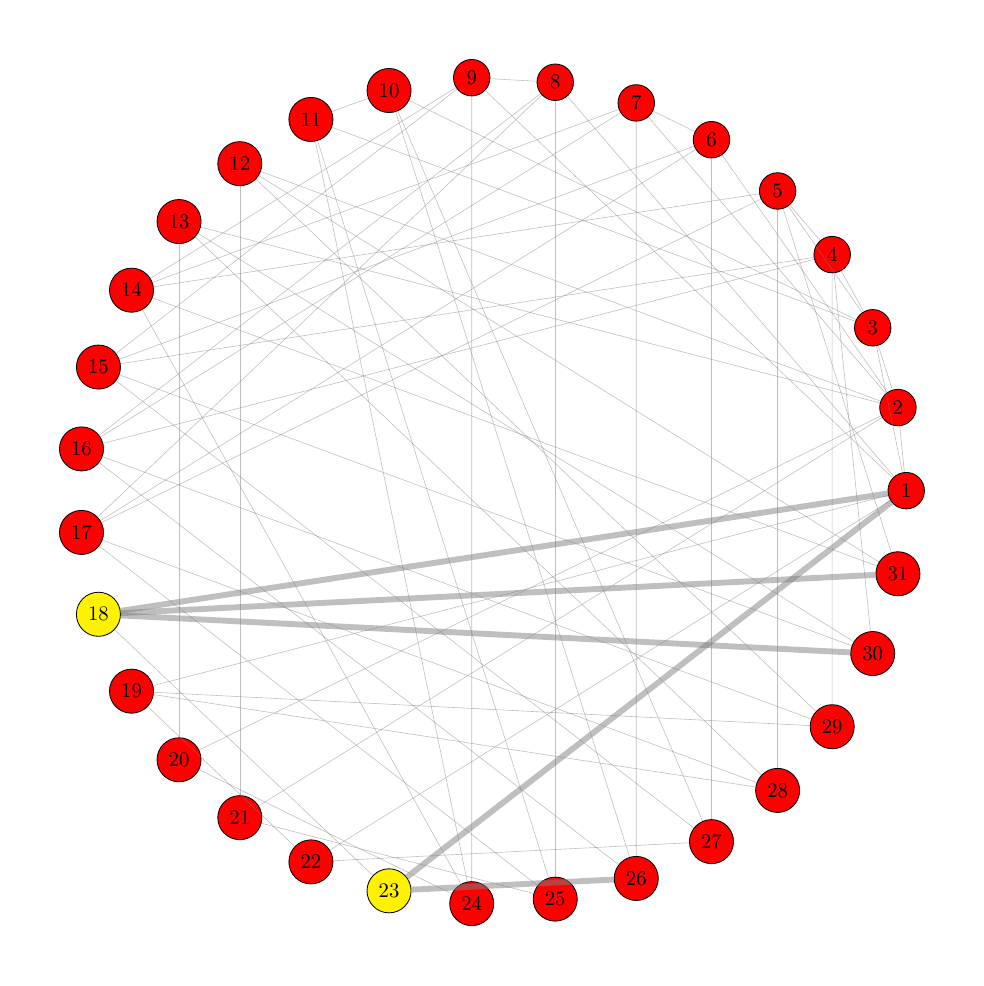}} \\[-10pt]
 \multirow{ 10}{*}{$4$} 
 & $10$ & $3,11$ & $(1, -e^{i\phi}/2, 0)$ \\
 & & & & \\[-12pt]
 & $12$ & $2,21$ & $(1, 0, -e^{i\gamma}/2)$\\
 & & & & \\[-12pt]
 & $20$ & $2,13$ & $(1, 0, -2e^{i\gamma})$\\ 
 & & & & \\[-12pt]
 & $22$ & $1,19$ & $(0, 1, -2e^{-i(\phi -\gamma)})$\\
 & & & & \\[-12pt]
 & $24$ & $9,11,14$ & $(1, -e^{i\phi}/2, e^{i\gamma}/2)$ \\
 & & & & \\[-12pt]
 & $26$ & $7,16$ & $(1, 2e^{i\phi}, -e^{i\gamma})$ \\
 & & & & \\[-12pt]
 & $27$ & $6,15$ & $(1, 2e^{i\phi}, e^{i\gamma})$ \\
 & & & & \\[-12pt]
 & $29$ & $4,16,19$ & $(1, -e^{i\phi}, 2e^{i\gamma})$ \\
 & & & & \\[-12pt]
 & $30$ & $4,13,15$ & $(1, -e^{i\phi}, -2e^{i\gamma})$\\
 & & & & \\[-12pt]
 & $31$ & $5,14$ & $(1, e^{i\phi}, 2e^{i\gamma})$ \\
 & & & & \\[-12pt]
 \hline 
 \multirow{ 2}{*}{$5$} 
 & $18$ & $1,30,31$ & $(0, 1, -e^{-i(\phi - \gamma)}/2)$ \\
 & $23$ & $1,26$ & $(0, 1, 2e^{-i(\phi - \gamma)})$ \\
 \hline \hline
 \end{tabular}
 \caption{$2$-neighbor bootstrap percolation process on the orthogonality graph of CK-31, starting with $A_0 = \{1,2,3,17\}$. For each round from $1$ to $5$, we indicate the newly infected vertices. We also indicate the initial infected set by Round 0.}
 \label{tab:CK31-percolation}
\end{table*}


\newpage 

\mbox{}



\begin{thebibliography}{84}%
\makeatletter
\providecommand \@ifxundefined [1]{%
 \@ifx{#1\undefined}
}%
\providecommand \@ifnum [1]{%
 \ifnum #1\expandafter \@firstoftwo
 \else \expandafter \@secondoftwo
 \fi
}%
\providecommand \@ifx [1]{%
 \ifx #1\expandafter \@firstoftwo
 \else \expandafter \@secondoftwo
 \fi
}%
\providecommand \natexlab [1]{#1}%
\providecommand \enquote  [1]{``#1''}%
\providecommand \bibnamefont  [1]{#1}%
\providecommand \bibfnamefont [1]{#1}%
\providecommand \citenamefont [1]{#1}%
\providecommand \href@noop [0]{\@secondoftwo}%
\providecommand \href [0]{\begingroup \@sanitize@url \@href}%
\providecommand \@href[1]{\@@startlink{#1}\@@href}%
\providecommand \@@href[1]{\endgroup#1\@@endlink}%
\providecommand \@sanitize@url [0]{\catcode `\\12\catcode `\$12\catcode `\&12\catcode `\#12\catcode `\^12\catcode `\_12\catcode `\%12\relax}%
\providecommand \@@startlink[1]{}%
\providecommand \@@endlink[0]{}%
\providecommand \url  [0]{\begingroup\@sanitize@url \@url }%
\providecommand \@url [1]{\endgroup\@href {#1}{\urlprefix }}%
\providecommand \urlprefix  [0]{URL }%
\providecommand \Eprint [0]{\href }%
\providecommand \doibase [0]{https://doi.org/}%
\providecommand \selectlanguage [0]{\@gobble}%
\providecommand \bibinfo  [0]{\@secondoftwo}%
\providecommand \bibfield  [0]{\@secondoftwo}%
\providecommand \translation [1]{[#1]}%
\providecommand \BibitemOpen [0]{}%
\providecommand \bibitemStop [0]{}%
\providecommand \bibitemNoStop [0]{.\EOS\space}%
\providecommand \EOS [0]{\spacefactor3000\relax}%
\providecommand \BibitemShut  [1]{\csname bibitem#1\endcsname}%
\let\auto@bib@innerbib\@empty
\bibitem [{\citenamefont {Kochen}\ and\ \citenamefont {Specker}(1967)}]{Kochen:1967JMM}%
  \BibitemOpen
  \bibfield  {author} {\bibinfo {author} {\bibfnamefont {S.}~\bibnamefont {Kochen}}\ and\ \bibinfo {author} {\bibfnamefont {E.~P.}\ \bibnamefont {Specker}},\ }\bibfield  {title} {\bibinfo {title} {The {P}roblem of {H}idden {V}ariables in {Q}uantum {M}echanics},\ }\href {https://doi.org/10.1512/iumj.1968.17.17004} {\bibfield  {journal} {\bibinfo  {journal} {J. Math. Mech.}\ }\textbf {\bibinfo {volume} {17}},\ \bibinfo {pages} {59} (\bibinfo {year} {1967})}\BibitemShut {NoStop}%
\bibitem [{\citenamefont {Stairs}(1983)}]{Stairs:1983PS}%
  \BibitemOpen
  \bibfield  {author} {\bibinfo {author} {\bibfnamefont {A.}~\bibnamefont {Stairs}},\ }\bibfield  {title} {\bibinfo {title} {Quantum logic, realism, and value definiteness},\ }\href {https://doi.org/10.1086/289140} {\bibfield  {journal} {\bibinfo  {journal} {Philos. Sci.}\ }\textbf {\bibinfo {volume} {50}},\ \bibinfo {pages} {578} (\bibinfo {year} {1983})}\BibitemShut {NoStop}%
\bibitem [{\citenamefont {Heywood}\ and\ \citenamefont {Redhead}(1983)}]{HR83}%
  \BibitemOpen
  \bibfield  {author} {\bibinfo {author} {\bibfnamefont {P.}~\bibnamefont {Heywood}}\ and\ \bibinfo {author} {\bibfnamefont {M.~L.~G.}\ \bibnamefont {Redhead}},\ }\bibfield  {title} {\bibinfo {title} {Nonlocality and the {K}ochen-{S}pecker paradox},\ }\href {https://doi.org/10.1007/BF00729511} {\bibfield  {journal} {\bibinfo  {journal} {Found. Phys.}\ }\textbf {\bibinfo {volume} {13}},\ \bibinfo {pages} {481} (\bibinfo {year} {1983})}\BibitemShut {NoStop}%
\bibitem [{\citenamefont {Cabello}(2001{\natexlab{a}})}]{Cabello:2001PRLa}%
  \BibitemOpen
  \bibfield  {author} {\bibinfo {author} {\bibfnamefont {A.}~\bibnamefont {Cabello}},\ }\bibfield  {title} {\bibinfo {title} {Bell's theorem without inequalities and without probabilities for two observers},\ }\href {https://doi.org/10.1103/PhysRevLett.86.1911} {\bibfield  {journal} {\bibinfo  {journal} {Phys. Rev. Lett.}\ }\textbf {\bibinfo {volume} {86}},\ \bibinfo {pages} {1911} (\bibinfo {year} {2001}{\natexlab{a}})}\BibitemShut {NoStop}%
\bibitem [{\citenamefont {Cabello}(2001{\natexlab{b}})}]{Cabello:2001PRLb}%
  \BibitemOpen
  \bibfield  {author} {\bibinfo {author} {\bibfnamefont {A.}~\bibnamefont {Cabello}},\ }\bibfield  {title} {\bibinfo {title} {``{A}ll versus {N}othing'' {I}nseparability for {T}wo {O}bservers},\ }\href {https://doi.org/10.1103/PhysRevLett.87.010403} {\bibfield  {journal} {\bibinfo  {journal} {Phys. Rev. Lett.}\ }\textbf {\bibinfo {volume} {87}},\ \bibinfo {pages} {010403} (\bibinfo {year} {2001}{\natexlab{b}})}\BibitemShut {NoStop}%
\bibitem [{\citenamefont {Cinelli}\ \emph {et~al.}(2005)\citenamefont {Cinelli}, \citenamefont {Barbieri}, \citenamefont {Perris}, \citenamefont {Mataloni},\ and\ \citenamefont {De~Martini}}]{CinelliPRL2005}%
  \BibitemOpen
  \bibfield  {author} {\bibinfo {author} {\bibfnamefont {C.}~\bibnamefont {Cinelli}}, \bibinfo {author} {\bibfnamefont {M.}~\bibnamefont {Barbieri}}, \bibinfo {author} {\bibfnamefont {R.}~\bibnamefont {Perris}}, \bibinfo {author} {\bibfnamefont {P.}~\bibnamefont {Mataloni}},\ and\ \bibinfo {author} {\bibfnamefont {F.}~\bibnamefont {De~Martini}},\ }\bibfield  {title} {\bibinfo {title} {{A}ll-{V}ersus-{N}othing {N}onlocality {T}est of {Q}uantum {M}echanics by {T}wo-{P}hoton {H}yperentanglement},\ }\href {https://doi.org/10.1103/PhysRevLett.95.240405} {\bibfield  {journal} {\bibinfo  {journal} {Phys. Rev. Lett.}\ }\textbf {\bibinfo {volume} {95}},\ \bibinfo {pages} {240405} (\bibinfo {year} {2005})}\BibitemShut {NoStop}%
\bibitem [{\citenamefont {Yang}\ \emph {et~al.}(2005)\citenamefont {Yang}, \citenamefont {Zhang}, \citenamefont {Zhang}, \citenamefont {Yin}, \citenamefont {Zhao}, \citenamefont {{\.Z}ukowski}, \citenamefont {Chen},\ and\ \citenamefont {Pan}}]{YangPRL2005}%
  \BibitemOpen
  \bibfield  {author} {\bibinfo {author} {\bibfnamefont {T.}~\bibnamefont {Yang}}, \bibinfo {author} {\bibfnamefont {Q.}~\bibnamefont {Zhang}}, \bibinfo {author} {\bibfnamefont {J.}~\bibnamefont {Zhang}}, \bibinfo {author} {\bibfnamefont {J.}~\bibnamefont {Yin}}, \bibinfo {author} {\bibfnamefont {Z.}~\bibnamefont {Zhao}}, \bibinfo {author} {\bibfnamefont {M.}~\bibnamefont {{\.Z}ukowski}}, \bibinfo {author} {\bibfnamefont {Z.-B.}\ \bibnamefont {Chen}},\ and\ \bibinfo {author} {\bibfnamefont {J.-W.}\ \bibnamefont {Pan}},\ }\bibfield  {title} {\bibinfo {title} {{A}ll-{V}ersus-{N}othing {V}iolation of {L}ocal {R}ealism by {T}wo-{P}hoton, {F}our-{D}imensional {E}ntanglement},\ }\href {https://doi.org/10.1103/PhysRevLett.95.240406} {\bibfield  {journal} {\bibinfo  {journal} {Phys. Rev. Lett.}\ }\textbf {\bibinfo {volume} {95}},\ \bibinfo {pages} {240406} (\bibinfo {year} {2005})}\BibitemShut {NoStop}%
\bibitem [{\citenamefont {Aolita}\ \emph {et~al.}(2012)\citenamefont {Aolita}, \citenamefont {Gallego}, \citenamefont {Ac{\'\i}n}, \citenamefont {Chiuri}, \citenamefont {Vallone}, \citenamefont {Mataloni},\ and\ \citenamefont {Cabello}}]{Aolita:2012PRA}%
  \BibitemOpen
  \bibfield  {author} {\bibinfo {author} {\bibfnamefont {L.}~\bibnamefont {Aolita}}, \bibinfo {author} {\bibfnamefont {R.}~\bibnamefont {Gallego}}, \bibinfo {author} {\bibfnamefont {A.}~\bibnamefont {Ac{\'\i}n}}, \bibinfo {author} {\bibfnamefont {A.}~\bibnamefont {Chiuri}}, \bibinfo {author} {\bibfnamefont {G.}~\bibnamefont {Vallone}}, \bibinfo {author} {\bibfnamefont {P.}~\bibnamefont {Mataloni}},\ and\ \bibinfo {author} {\bibfnamefont {A.}~\bibnamefont {Cabello}},\ }\bibfield  {title} {\bibinfo {title} {Fully nonlocal quantum correlations},\ }\href {https://doi.org/10.1103/PhysRevA.85.032107} {\bibfield  {journal} {\bibinfo  {journal} {Phys. Rev. A}\ }\textbf {\bibinfo {volume} {85}},\ \bibinfo {pages} {032107} (\bibinfo {year} {2012})}\BibitemShut {NoStop}%
\bibitem [{\citenamefont {Cabello}(2021)}]{Cabello:2021PRL}%
  \BibitemOpen
  \bibfield  {author} {\bibinfo {author} {\bibfnamefont {A.}~\bibnamefont {Cabello}},\ }\bibfield  {title} {\bibinfo {title} {Converting contextuality into nonlocality},\ }\href {https://doi.org/10.1103/PhysRevLett.127.070401} {\bibfield  {journal} {\bibinfo  {journal} {Phys. Rev. Lett.}\ }\textbf {\bibinfo {volume} {127}},\ \bibinfo {pages} {070401} (\bibinfo {year} {2021})}\BibitemShut {NoStop}%
\bibitem [{\citenamefont {Xu}\ \emph {et~al.}(2022)\citenamefont {Xu}, \citenamefont {Zhen}, \citenamefont {Yang}, \citenamefont {Cheng}, \citenamefont {Ren}, \citenamefont {Chen}, \citenamefont {Wang},\ and\ \citenamefont {Wang}}]{Xu:2022PRL}%
  \BibitemOpen
  \bibfield  {author} {\bibinfo {author} {\bibfnamefont {J.-M.}\ \bibnamefont {Xu}}, \bibinfo {author} {\bibfnamefont {Y.-Z.}\ \bibnamefont {Zhen}}, \bibinfo {author} {\bibfnamefont {Y.-X.}\ \bibnamefont {Yang}}, \bibinfo {author} {\bibfnamefont {Z.-M.}\ \bibnamefont {Cheng}}, \bibinfo {author} {\bibfnamefont {Z.-C.}\ \bibnamefont {Ren}}, \bibinfo {author} {\bibfnamefont {K.}~\bibnamefont {Chen}}, \bibinfo {author} {\bibfnamefont {X.-L.}\ \bibnamefont {Wang}},\ and\ \bibinfo {author} {\bibfnamefont {H.-T.}\ \bibnamefont {Wang}},\ }\bibfield  {title} {\bibinfo {title} {Experimental {D}emonstration of {Q}uantum {P}seudotelepathy},\ }\href {https://doi.org/10.1103/PhysRevLett.129.050402} {\bibfield  {journal} {\bibinfo  {journal} {Phys. Rev. Lett.}\ }\textbf {\bibinfo {volume} {129}},\ \bibinfo {pages} {050402} (\bibinfo {year} {2022})}\BibitemShut {NoStop}%
\bibitem [{\citenamefont {Sheng}\ \emph {et~al.}(2025)\citenamefont {Sheng}, \citenamefont {Zhang},\ and\ \citenamefont {Chen}}]{Sheng:2025PRL}%
  \BibitemOpen
  \bibfield  {author} {\bibinfo {author} {\bibfnamefont {J.}~\bibnamefont {Sheng}}, \bibinfo {author} {\bibfnamefont {D.}~\bibnamefont {Zhang}},\ and\ \bibinfo {author} {\bibfnamefont {L.}~\bibnamefont {Chen}},\ }\bibfield  {title} {\bibinfo {title} {Orbital angular momentum experiment converting contextuality into nonlocality},\ }\href {https://doi.org/10.1103/PhysRevLett.134.010203} {\bibfield  {journal} {\bibinfo  {journal} {Phys. Rev. Lett.}\ }\textbf {\bibinfo {volume} {134}},\ \bibinfo {pages} {010203} (\bibinfo {year} {2025})}\BibitemShut {NoStop}%
\bibitem [{\citenamefont {Cabello}(2008)}]{Cabello:2008PRL}%
  \BibitemOpen
  \bibfield  {author} {\bibinfo {author} {\bibfnamefont {A.}~\bibnamefont {Cabello}},\ }\bibfield  {title} {\bibinfo {title} {Experimentally testable state-independent quantum contextuality},\ }\href {https://doi.org/10.1103/PhysRevLett.101.210401} {\bibfield  {journal} {\bibinfo  {journal} {Phys. Rev. Lett.}\ }\textbf {\bibinfo {volume} {101}},\ \bibinfo {pages} {210401} (\bibinfo {year} {2008})}\BibitemShut {NoStop}%
\bibitem [{\citenamefont {Badzia\c{g}}\ \emph {et~al.}(2009)\citenamefont {Badzia\c{g}}, \citenamefont {Bengtsson}, \citenamefont {Cabello},\ and\ \citenamefont {Pitowsky}}]{Badziag:2009PRL}%
  \BibitemOpen
  \bibfield  {author} {\bibinfo {author} {\bibfnamefont {P.}~\bibnamefont {Badzia\c{g}}}, \bibinfo {author} {\bibfnamefont {I.}~\bibnamefont {Bengtsson}}, \bibinfo {author} {\bibfnamefont {A.}~\bibnamefont {Cabello}},\ and\ \bibinfo {author} {\bibfnamefont {I.}~\bibnamefont {Pitowsky}},\ }\bibfield  {title} {\bibinfo {title} {Universality of state-independent violation of correlation inequalities for noncontextual theories},\ }\href {https://doi.org/10.1103/PhysRevLett.103.050401} {\bibfield  {journal} {\bibinfo  {journal} {Phys. Rev. Lett.}\ }\textbf {\bibinfo {volume} {103}},\ \bibinfo {pages} {050401} (\bibinfo {year} {2009})}\BibitemShut {NoStop}%
\bibitem [{\citenamefont {Kirchmair}\ \emph {et~al.}(2009)\citenamefont {Kirchmair}, \citenamefont {Z{\"a}hringer}, \citenamefont {Gerritsma}, \citenamefont {Kleinmann}, \citenamefont {G{\"u}hne}, \citenamefont {Cabello}, \citenamefont {Blatt},\ and\ \citenamefont {Roos}}]{Kirchmair:2009NAT}%
  \BibitemOpen
  \bibfield  {author} {\bibinfo {author} {\bibfnamefont {G.}~\bibnamefont {Kirchmair}}, \bibinfo {author} {\bibfnamefont {F.}~\bibnamefont {Z{\"a}hringer}}, \bibinfo {author} {\bibfnamefont {R.}~\bibnamefont {Gerritsma}}, \bibinfo {author} {\bibfnamefont {M.}~\bibnamefont {Kleinmann}}, \bibinfo {author} {\bibfnamefont {O.}~\bibnamefont {G{\"u}hne}}, \bibinfo {author} {\bibfnamefont {A.}~\bibnamefont {Cabello}}, \bibinfo {author} {\bibfnamefont {R.}~\bibnamefont {Blatt}},\ and\ \bibinfo {author} {\bibfnamefont {C.~F.}\ \bibnamefont {Roos}},\ }\bibfield  {title} {\bibinfo {title} {State-independent experimental test of quantum contextuality},\ }\href {https://doi.org/10.1038/nature08172} {\bibfield  {journal} {\bibinfo  {journal} {Nature}\ }\textbf {\bibinfo {volume} {460}},\ \bibinfo {pages} {494} (\bibinfo {year} {2009})}\BibitemShut {NoStop}%
\bibitem [{\citenamefont {Amselem}\ \emph {et~al.}(2009)\citenamefont {Amselem}, \citenamefont {R{\aa}dmark}, \citenamefont {Bourennane},\ and\ \citenamefont {Cabello}}]{Amselem:2009PRL}%
  \BibitemOpen
  \bibfield  {author} {\bibinfo {author} {\bibfnamefont {E.}~\bibnamefont {Amselem}}, \bibinfo {author} {\bibfnamefont {M.}~\bibnamefont {R{\aa}dmark}}, \bibinfo {author} {\bibfnamefont {M.}~\bibnamefont {Bourennane}},\ and\ \bibinfo {author} {\bibfnamefont {A.}~\bibnamefont {Cabello}},\ }\bibfield  {title} {\bibinfo {title} {State-independent quantum contextuality with single photons},\ }\href {https://doi.org/10.1103/PhysRevLett.103.160405} {\bibfield  {journal} {\bibinfo  {journal} {Phys. Rev. Lett.}\ }\textbf {\bibinfo {volume} {103}},\ \bibinfo {pages} {160405} (\bibinfo {year} {2009})}\BibitemShut {NoStop}%
\bibitem [{\citenamefont {D'Ambrosio}\ \emph {et~al.}(2013)\citenamefont {D'Ambrosio}, \citenamefont {Herbauts}, \citenamefont {Amselem}, \citenamefont {Nagali}, \citenamefont {Bourennane}, \citenamefont {Sciarrino},\ and\ \citenamefont {Cabello}}]{D'Ambrosio:2013PRX}%
  \BibitemOpen
  \bibfield  {author} {\bibinfo {author} {\bibfnamefont {V.}~\bibnamefont {D'Ambrosio}}, \bibinfo {author} {\bibfnamefont {I.}~\bibnamefont {Herbauts}}, \bibinfo {author} {\bibfnamefont {E.}~\bibnamefont {Amselem}}, \bibinfo {author} {\bibfnamefont {E.}~\bibnamefont {Nagali}}, \bibinfo {author} {\bibfnamefont {M.}~\bibnamefont {Bourennane}}, \bibinfo {author} {\bibfnamefont {F.}~\bibnamefont {Sciarrino}},\ and\ \bibinfo {author} {\bibfnamefont {A.}~\bibnamefont {Cabello}},\ }\bibfield  {title} {\bibinfo {title} {Experimental implementation of a {K}ochen-{S}pecker set of quantum tests},\ }\href {https://doi.org/10.1103/PhysRevX.3.011012} {\bibfield  {journal} {\bibinfo  {journal} {Phys. Rev. X}\ }\textbf {\bibinfo {volume} {3}},\ \bibinfo {pages} {011012} (\bibinfo {year} {2013})}\BibitemShut {NoStop}%
\bibitem [{\citenamefont {Yu}\ and\ \citenamefont {Oh}(2012)}]{Yu:2012PRL}%
  \BibitemOpen
  \bibfield  {author} {\bibinfo {author} {\bibfnamefont {S.}~\bibnamefont {Yu}}\ and\ \bibinfo {author} {\bibfnamefont {C.~H.}\ \bibnamefont {Oh}},\ }\bibfield  {title} {\bibinfo {title} {State-independent proof of {K}ochen-{S}pecker theorem with 13 rays},\ }\href {https://doi.org/10.1103/PhysRevLett.108.030402} {\bibfield  {journal} {\bibinfo  {journal} {Phys. Rev. Lett.}\ }\textbf {\bibinfo {volume} {108}},\ \bibinfo {pages} {030402} (\bibinfo {year} {2012})}\BibitemShut {NoStop}%
\bibitem [{\citenamefont {Bengtsson}\ \emph {et~al.}(2012)\citenamefont {Bengtsson}, \citenamefont {Blanchfield},\ and\ \citenamefont {Cabello}}]{Bengtsson:2012PLA}%
  \BibitemOpen
  \bibfield  {author} {\bibinfo {author} {\bibfnamefont {I.}~\bibnamefont {Bengtsson}}, \bibinfo {author} {\bibfnamefont {K.}~\bibnamefont {Blanchfield}},\ and\ \bibinfo {author} {\bibfnamefont {A.}~\bibnamefont {Cabello}},\ }\bibfield  {title} {\bibinfo {title} {A {K}ochen--{S}pecker inequality from a {SIC}},\ }\href {https://doi.org/10.1016/j.physleta.2011.12.011} {\bibfield  {journal} {\bibinfo  {journal} {Phys. Lett. A}\ }\textbf {\bibinfo {volume} {376}},\ \bibinfo {pages} {374} (\bibinfo {year} {2012})}\BibitemShut {NoStop}%
\bibitem [{\citenamefont {Kleinmann}\ \emph {et~al.}(2012)\citenamefont {Kleinmann}, \citenamefont {Budroni}, \citenamefont {Larsson}, \citenamefont {G{\"u}hne},\ and\ \citenamefont {Cabello}}]{Kleinmann:2012PRL}%
  \BibitemOpen
  \bibfield  {author} {\bibinfo {author} {\bibfnamefont {M.}~\bibnamefont {Kleinmann}}, \bibinfo {author} {\bibfnamefont {C.}~\bibnamefont {Budroni}}, \bibinfo {author} {\bibfnamefont {J.-{\AA}.}\ \bibnamefont {Larsson}}, \bibinfo {author} {\bibfnamefont {O.}~\bibnamefont {G{\"u}hne}},\ and\ \bibinfo {author} {\bibfnamefont {A.}~\bibnamefont {Cabello}},\ }\bibfield  {title} {\bibinfo {title} {Optimal inequalities for state-independent contextuality},\ }\href {https://doi.org/10.1103/PhysRevLett.109.250402} {\bibfield  {journal} {\bibinfo  {journal} {Phys. Rev. Lett.}\ }\textbf {\bibinfo {volume} {109}},\ \bibinfo {pages} {250402} (\bibinfo {year} {2012})}\BibitemShut {NoStop}%
\bibitem [{\citenamefont {Cabello}\ \emph {et~al.}(2016)\citenamefont {Cabello}, \citenamefont {Kleinmann},\ and\ \citenamefont {Portillo}}]{Cabello:2016JPA}%
  \BibitemOpen
  \bibfield  {author} {\bibinfo {author} {\bibfnamefont {A.}~\bibnamefont {Cabello}}, \bibinfo {author} {\bibfnamefont {M.}~\bibnamefont {Kleinmann}},\ and\ \bibinfo {author} {\bibfnamefont {J.~R.}\ \bibnamefont {Portillo}},\ }\bibfield  {title} {\bibinfo {title} {Quantum state-independent contextuality requires 13 rays},\ }\href {https://doi.org/10.1088/1751-8113/49/38/38LT01} {\bibfield  {journal} {\bibinfo  {journal} {J. Phys. A: Math. Theor.}\ }\textbf {\bibinfo {volume} {49}},\ \bibinfo {pages} {38LT01} (\bibinfo {year} {2016})}\BibitemShut {NoStop}%
\bibitem [{\citenamefont {Kirchweger}\ \emph {et~al.}(2023)\citenamefont {Kirchweger}, \citenamefont {Peitl},\ and\ \citenamefont {Szeider}}]{Kirchweger:2023}%
  \BibitemOpen
  \bibfield  {author} {\bibinfo {author} {\bibfnamefont {M.}~\bibnamefont {Kirchweger}}, \bibinfo {author} {\bibfnamefont {T.}~\bibnamefont {Peitl}},\ and\ \bibinfo {author} {\bibfnamefont {S.}~\bibnamefont {Szeider}},\ }\bibfield  {title} {\bibinfo {title} {Co-certificate learning with {SAT} modulo symmetries},\ }in\ \href {https://doi.org/10.24963/ijcai.2023/216} {\emph {\bibinfo {booktitle} {Proceedings of the Thirty-Second International Joint Conference on Artificial Intelligence, {IJCAI-23}}}},\ \bibinfo {editor} {edited by\ \bibinfo {editor} {\bibfnamefont {E.}~\bibnamefont {Elkind}}}\ (\bibinfo  {publisher} {International Joint Conferences on Artificial Intelligence Organization},\ \bibinfo {year} {2023})\ pp.\ \bibinfo {pages} {1944--1953}\BibitemShut {NoStop}%
\bibitem [{\citenamefont {Li}\ \emph {et~al.}(2024)\citenamefont {Li}, \citenamefont {Bright},\ and\ \citenamefont {Ganesh}}]{Li:2024}%
  \BibitemOpen
  \bibfield  {author} {\bibinfo {author} {\bibfnamefont {Z.}~\bibnamefont {Li}}, \bibinfo {author} {\bibfnamefont {C.}~\bibnamefont {Bright}},\ and\ \bibinfo {author} {\bibfnamefont {V.}~\bibnamefont {Ganesh}},\ }\bibfield  {title} {\bibinfo {title} {A {SAT} solver + computer algebra attack on the minimum {K}ochen–{S}pecker problem},\ }in\ \href {https://doi.org/10.24963/ijcai.2024/210} {\emph {\bibinfo {booktitle} {Proceedings of the Thirty-Third International Joint Conference on Artificial Intelligence, {IJCAI-24}}}},\ \bibinfo {editor} {edited by\ \bibinfo {editor} {\bibfnamefont {K.}~\bibnamefont {Larson}}}\ (\bibinfo  {publisher} {International Joint Conferences on Artificial Intelligence Organization},\ \bibinfo {year} {2024})\ pp.\ \bibinfo {pages} {1898--1906}\BibitemShut {NoStop}%
\bibitem [{\citenamefont {Xu}\ \emph {et~al.}(2020)\citenamefont {Xu}, \citenamefont {Chen},\ and\ \citenamefont {G{\"u}hne}}]{Xu:2020PRL}%
  \BibitemOpen
  \bibfield  {author} {\bibinfo {author} {\bibfnamefont {Z.-P.}\ \bibnamefont {Xu}}, \bibinfo {author} {\bibfnamefont {J.-L.}\ \bibnamefont {Chen}},\ and\ \bibinfo {author} {\bibfnamefont {O.}~\bibnamefont {G{\"u}hne}},\ }\bibfield  {title} {\bibinfo {title} {Proof of the {P}eres conjecture for contextuality},\ }\href {https://doi.org/10.1103/PhysRevLett.124.230401} {\bibfield  {journal} {\bibinfo  {journal} {Phys. Rev. Lett.}\ }\textbf {\bibinfo {volume} {124}},\ \bibinfo {pages} {230401} (\bibinfo {year} {2020})}\BibitemShut {NoStop}%
\bibitem [{\citenamefont {Cabello}\ \emph {et~al.}(1996)\citenamefont {Cabello}, \citenamefont {Estebaranz},\ and\ \citenamefont {Garc{\'\i}a-Alcaine}}]{Cabello:1996PLA}%
  \BibitemOpen
  \bibfield  {author} {\bibinfo {author} {\bibfnamefont {A.}~\bibnamefont {Cabello}}, \bibinfo {author} {\bibfnamefont {J.~M.}\ \bibnamefont {Estebaranz}},\ and\ \bibinfo {author} {\bibfnamefont {G.}~\bibnamefont {Garc{\'\i}a-Alcaine}},\ }\bibfield  {title} {\bibinfo {title} {Bell-{K}ochen-{S}pecker theorem: A proof with 18 vectors},\ }\href {https://doi.org/10.1016/0375-9601(96)00134-X} {\bibfield  {journal} {\bibinfo  {journal} {Phys. Lett. A}\ }\textbf {\bibinfo {volume} {212}},\ \bibinfo {pages} {183} (\bibinfo {year} {1996})}\BibitemShut {NoStop}%
\bibitem [{\citenamefont {Xu}\ \emph {et~al.}(2024)\citenamefont {Xu}, \citenamefont {Saha}, \citenamefont {Bharti},\ and\ \citenamefont {Cabello}}]{Xu:2024PRL}%
  \BibitemOpen
  \bibfield  {author} {\bibinfo {author} {\bibfnamefont {Z.-P.}\ \bibnamefont {Xu}}, \bibinfo {author} {\bibfnamefont {D.}~\bibnamefont {Saha}}, \bibinfo {author} {\bibfnamefont {K.}~\bibnamefont {Bharti}},\ and\ \bibinfo {author} {\bibfnamefont {A.}~\bibnamefont {Cabello}},\ }\bibfield  {title} {\bibinfo {title} {Certifying sets of quantum observables with any full-rank state},\ }\href {https://doi.org/10.1103/PhysRevLett.132.140201} {\bibfield  {journal} {\bibinfo  {journal} {Phys. Rev. Lett.}\ }\textbf {\bibinfo {volume} {132}},\ \bibinfo {pages} {140201} (\bibinfo {year} {2024})}\BibitemShut {NoStop}%
\bibitem [{\citenamefont {Liu}\ \emph {et~al.}(2024)\citenamefont {Liu}, \citenamefont {Chung}, \citenamefont {Cruzeiro}, \citenamefont {Gonzales-Ureta}, \citenamefont {Ramanathan},\ and\ \citenamefont {Cabello}}]{Liu:2023XXX}%
  \BibitemOpen
  \bibfield  {author} {\bibinfo {author} {\bibfnamefont {Y.}~\bibnamefont {Liu}}, \bibinfo {author} {\bibfnamefont {H.~Y.}\ \bibnamefont {Chung}}, \bibinfo {author} {\bibfnamefont {E.~Z.}\ \bibnamefont {Cruzeiro}}, \bibinfo {author} {\bibfnamefont {J.~R.}\ \bibnamefont {Gonzales-Ureta}}, \bibinfo {author} {\bibfnamefont {R.}~\bibnamefont {Ramanathan}},\ and\ \bibinfo {author} {\bibfnamefont {A.}~\bibnamefont {Cabello}},\ }\bibfield  {title} {\bibinfo {title} {Equivalence between face nonsignaling correlations, full nonlocality, all-versus-nothing proofs, and pseudotelepathy},\ }\href {https://doi.org/10.1103/PhysRevResearch.6.L042035} {\bibfield  {journal} {\bibinfo  {journal} {Phys. Rev. Res.}\ }\textbf {\bibinfo {volume} {6}},\ \bibinfo {pages} {L042035} (\bibinfo {year} {2024})}\BibitemShut {NoStop}%
\bibitem [{\citenamefont {Cabello}(2025)}]{Cabello:2025PRL}%
  \BibitemOpen
  \bibfield  {author} {\bibinfo {author} {\bibfnamefont {A.}~\bibnamefont {Cabello}},\ }\bibfield  {title} {\bibinfo {title} {Simplest bipartite perfect quantum strategies},\ }\href {https://doi.org/10.1103/PhysRevLett.134.010201} {\bibfield  {journal} {\bibinfo  {journal} {Phys. Rev. Lett.}\ }\textbf {\bibinfo {volume} {134}},\ \bibinfo {pages} {010201} (\bibinfo {year} {2025})}\BibitemShut {NoStop}%
\bibitem [{\citenamefont {Goh}\ \emph {et~al.}(2018)\citenamefont {Goh}, \citenamefont {Kaniewski}, \citenamefont {Wolfe}, \citenamefont {V\'ertesi}, \citenamefont {Wu}, \citenamefont {Cai}, \citenamefont {Liang},\ and\ \citenamefont {Scarani}}]{Goh:2018PRA}%
  \BibitemOpen
  \bibfield  {author} {\bibinfo {author} {\bibfnamefont {K.~T.}\ \bibnamefont {Goh}}, \bibinfo {author} {\bibfnamefont {J.}~\bibnamefont {Kaniewski}}, \bibinfo {author} {\bibfnamefont {E.}~\bibnamefont {Wolfe}}, \bibinfo {author} {\bibfnamefont {T.}~\bibnamefont {V\'ertesi}}, \bibinfo {author} {\bibfnamefont {X.}~\bibnamefont {Wu}}, \bibinfo {author} {\bibfnamefont {Y.}~\bibnamefont {Cai}}, \bibinfo {author} {\bibfnamefont {Y.-C.}\ \bibnamefont {Liang}},\ and\ \bibinfo {author} {\bibfnamefont {V.}~\bibnamefont {Scarani}},\ }\bibfield  {title} {\bibinfo {title} {Geometry of the set of quantum correlations},\ }\href {https://doi.org/10.1103/PhysRevA.97.022104} {\bibfield  {journal} {\bibinfo  {journal} {Phys. Rev. A}\ }\textbf {\bibinfo {volume} {97}},\ \bibinfo {pages} {022104} (\bibinfo {year} {2018})}\BibitemShut {NoStop}%
\bibitem [{\citenamefont {Elitzur}\ \emph {et~al.}(1992)\citenamefont {Elitzur}, \citenamefont {Popescu},\ and\ \citenamefont {Rohrlich}}]{Elitzur:1992PLA}%
  \BibitemOpen
  \bibfield  {author} {\bibinfo {author} {\bibfnamefont {A.~C.}\ \bibnamefont {Elitzur}}, \bibinfo {author} {\bibfnamefont {S.}~\bibnamefont {Popescu}},\ and\ \bibinfo {author} {\bibfnamefont {D.}~\bibnamefont {Rohrlich}},\ }\bibfield  {title} {\bibinfo {title} {Quantum nonlocality for each pair in an ensemble},\ }\href {https://doi.org/10.1016/0375-9601(92)90952-I} {\bibfield  {journal} {\bibinfo  {journal} {Phys. Lett. A}\ }\textbf {\bibinfo {volume} {162}},\ \bibinfo {pages} {25} (\bibinfo {year} {1992})}\BibitemShut {NoStop}%
\bibitem [{\citenamefont {Ji}\ \emph {et~al.}(2021)\citenamefont {Ji}, \citenamefont {Natarajan}, \citenamefont {Vidick}, \citenamefont {Wright},\ and\ \citenamefont {Yuen}}]{Ji:2021CACM}%
  \BibitemOpen
  \bibfield  {author} {\bibinfo {author} {\bibfnamefont {Z.}~\bibnamefont {Ji}}, \bibinfo {author} {\bibfnamefont {A.}~\bibnamefont {Natarajan}}, \bibinfo {author} {\bibfnamefont {T.}~\bibnamefont {Vidick}}, \bibinfo {author} {\bibfnamefont {J.}~\bibnamefont {Wright}},\ and\ \bibinfo {author} {\bibfnamefont {H.}~\bibnamefont {Yuen}},\ }\bibfield  {title} {\bibinfo {title} {{MIP*=RE}},\ }\href {https://doi.org/10.1145/3485628} {\bibfield  {journal} {\bibinfo  {journal} {Comm. ACM}\ }\textbf {\bibinfo {volume} {64}},\ \bibinfo {pages} {131} (\bibinfo {year} {2021})}\BibitemShut {NoStop}%
\bibitem [{\citenamefont {Bravyi}\ \emph {et~al.}(2018)\citenamefont {Bravyi}, \citenamefont {Gosset},\ and\ \citenamefont {K{\"o}nig}}]{Bravyi:2018SCI}%
  \BibitemOpen
  \bibfield  {author} {\bibinfo {author} {\bibfnamefont {S.}~\bibnamefont {Bravyi}}, \bibinfo {author} {\bibfnamefont {D.}~\bibnamefont {Gosset}},\ and\ \bibinfo {author} {\bibfnamefont {R.}~\bibnamefont {K{\"o}nig}},\ }\bibfield  {title} {\bibinfo {title} {Quantum advantage with shallow circuits},\ }\href {https://doi.org/10.1126/science.aar3106} {\bibfield  {journal} {\bibinfo  {journal} {Science}\ }\textbf {\bibinfo {volume} {362}},\ \bibinfo {pages} {308} (\bibinfo {year} {2018})}\BibitemShut {NoStop}%
\bibitem [{\citenamefont {Mayers}\ and\ \citenamefont {Yao}(2004)}]{Yao_self}%
  \BibitemOpen
  \bibfield  {author} {\bibinfo {author} {\bibfnamefont {D.}~\bibnamefont {Mayers}}\ and\ \bibinfo {author} {\bibfnamefont {A.}~\bibnamefont {Yao}},\ }\bibfield  {title} {\bibinfo {title} {Self testing quantum apparatus},\ }\href {http://dl.acm.org/citation.cfm?id=2011827.2011830} {\bibfield  {journal} {\bibinfo  {journal} {Quantum Info. Comput.}\ }\textbf {\bibinfo {volume} {4}},\ \bibinfo {pages} {273} (\bibinfo {year} {2004})}\BibitemShut {NoStop}%
\bibitem [{\citenamefont {{\v{S}}upi{\'{c}}}\ and\ \citenamefont {Bowles}(2020)}]{Supic:2020Q}%
  \BibitemOpen
  \bibfield  {author} {\bibinfo {author} {\bibfnamefont {I.}~\bibnamefont {{\v{S}}upi{\'{c}}}}\ and\ \bibinfo {author} {\bibfnamefont {J.}~\bibnamefont {Bowles}},\ }\bibfield  {title} {\bibinfo {title} {Self-testing of quantum systems: A review},\ }\href {https://doi.org/10.22331/q-2020-09-30-337} {\bibfield  {journal} {\bibinfo  {journal} {{Quantum}}\ }\textbf {\bibinfo {volume} {4}},\ \bibinfo {pages} {337} (\bibinfo {year} {2020})}\BibitemShut {NoStop}%
\bibitem [{\citenamefont {Cabello}(2002)}]{Cabello:2002PRL}%
  \BibitemOpen
  \bibfield  {author} {\bibinfo {author} {\bibfnamefont {A.}~\bibnamefont {Cabello}},\ }\bibfield  {title} {\bibinfo {title} {{$N$}-particle {$N$}-level singlet states: Some properties and applications},\ }\href {https://doi.org/10.1103/PhysRevLett.89.100402} {\bibfield  {journal} {\bibinfo  {journal} {Phys. Rev. Lett.}\ }\textbf {\bibinfo {volume} {89}},\ \bibinfo {pages} {100402} (\bibinfo {year} {2002})}\BibitemShut {NoStop}%
\bibitem [{\citenamefont {Cabello}(2003)}]{Cabello:2003JMP}%
  \BibitemOpen
  \bibfield  {author} {\bibinfo {author} {\bibfnamefont {A.}~\bibnamefont {Cabello}},\ }\bibfield  {title} {\bibinfo {title} {Supersinglets},\ }\href {https://doi.org/10.1080/09500340308234551} {\bibfield  {journal} {\bibinfo  {journal} {J. Mod. Opt.}\ }\textbf {\bibinfo {volume} {50}},\ \bibinfo {pages} {1049} (\bibinfo {year} {2003})}\BibitemShut {NoStop}%
\bibitem [{\citenamefont {Ilo-Okeke}\ \emph {et~al.}(2022)\citenamefont {Ilo-Okeke}, \citenamefont {Ji}, \citenamefont {Chen}, \citenamefont {Mao}, \citenamefont {Kondappan}, \citenamefont {Ivannikov}, \citenamefont {Xiao},\ and\ \citenamefont {Byrnes}}]{PhysRevA.106.033314}%
  \BibitemOpen
  \bibfield  {author} {\bibinfo {author} {\bibfnamefont {E.~O.}\ \bibnamefont {Ilo-Okeke}}, \bibinfo {author} {\bibfnamefont {Y.}~\bibnamefont {Ji}}, \bibinfo {author} {\bibfnamefont {P.}~\bibnamefont {Chen}}, \bibinfo {author} {\bibfnamefont {Y.}~\bibnamefont {Mao}}, \bibinfo {author} {\bibfnamefont {M.}~\bibnamefont {Kondappan}}, \bibinfo {author} {\bibfnamefont {V.}~\bibnamefont {Ivannikov}}, \bibinfo {author} {\bibfnamefont {Y.}~\bibnamefont {Xiao}},\ and\ \bibinfo {author} {\bibfnamefont {T.}~\bibnamefont {Byrnes}},\ }\bibfield  {title} {\bibinfo {title} {Deterministic preparation of supersinglets with collective spin projections},\ }\href {https://doi.org/10.1103/PhysRevA.106.033314} {\bibfield  {journal} {\bibinfo  {journal} {Phys. Rev. A}\ }\textbf {\bibinfo {volume} {106}},\ \bibinfo {pages} {033314} (\bibinfo {year} {2022})}\BibitemShut {NoStop}%
\bibitem [{\citenamefont {Saha}\ and\ \citenamefont {Cabello}(2025)}]{Saha:2025XXX}%
  \BibitemOpen
  \bibfield  {author} {\bibinfo {author} {\bibfnamefont {D.}~\bibnamefont {Saha}}\ and\ \bibinfo {author} {\bibfnamefont {A.}~\bibnamefont {Cabello}},\ }\bibfield  {title} {\bibinfo {title} {Self-testing supersinglets with perfect quantum strategies},\ }\href {https://doi.org/10.1103/3ymw-sn2z} {\bibfield  {journal} {\bibinfo  {journal} {Phys. Rev. A}\ }\textbf {\bibinfo {volume} {112}},\ \bibinfo {pages} {012606} (\bibinfo {year} {2025})}\BibitemShut {NoStop}%
\bibitem [{\citenamefont {Manin}(1981)}]{Manin:1981}%
  \BibitemOpen
  \bibfield  {author} {\bibinfo {author} {\bibfnamefont {Y.~I.}\ \bibnamefont {Manin}},\ }\href@noop {} {\emph {\bibinfo {title} {The Demonstrable and Undemonstrable}}}\ (\bibinfo  {publisher} {Mir},\ \bibinfo {address} {Moscow, USSR},\ \bibinfo {year} {1981})\BibitemShut {NoStop}%
\bibitem [{\citenamefont {Redhead}(1987)}]{Redhead:1987}%
  \BibitemOpen
  \bibfield  {author} {\bibinfo {author} {\bibfnamefont {M.~L.~G.}\ \bibnamefont {Redhead}},\ }\href@noop {} {\emph {\bibinfo {title} {Incompleteness, Nonlocality, and Realism}}}\ (\bibinfo  {publisher} {Oxford University Press},\ \bibinfo {address} {New York},\ \bibinfo {year} {1987})\BibitemShut {NoStop}%
\bibitem [{\citenamefont {Halvorson}(2011)}]{Halvorson:2011}%
  \BibitemOpen
  \bibinfo {editor} {\bibfnamefont {H.}~\bibnamefont {Halvorson}},\ ed.,\ \href@noop {} {\emph {\bibinfo {title} {Deep Beauty: Understanding the Quantum World through Mathematical Innovation}}}\ (\bibinfo  {publisher} {Cambridge University Press},\ \bibinfo {year} {2011})\BibitemShut {NoStop}%
\bibitem [{\citenamefont {Conway}\ and\ \citenamefont {Kochen}(2006)}]{CK06}%
  \BibitemOpen
  \bibfield  {author} {\bibinfo {author} {\bibfnamefont {J.}~\bibnamefont {Conway}}\ and\ \bibinfo {author} {\bibfnamefont {S.}~\bibnamefont {Kochen}},\ }\bibfield  {title} {\bibinfo {title} {The free will theorem},\ }\href {https://doi.org/10.1007/s10701-006-9068-6} {\bibfield  {journal} {\bibinfo  {journal} {Found. Phys.}\ }\textbf {\bibinfo {volume} {36}},\ \bibinfo {pages} {1441} (\bibinfo {year} {2006})}\BibitemShut {NoStop}%
\bibitem [{\citenamefont {Conway}\ and\ \citenamefont {Kochen}(2009)}]{CK09}%
  \BibitemOpen
  \bibfield  {author} {\bibinfo {author} {\bibfnamefont {J.~H.}\ \bibnamefont {Conway}}\ and\ \bibinfo {author} {\bibfnamefont {S.}~\bibnamefont {Kochen}},\ }\bibfield  {title} {\bibinfo {title} {The strong free will theorem},\ }\href@noop {} {\bibfield  {journal} {\bibinfo  {journal} {Notices Amer. Math. Soc.}\ }\textbf {\bibinfo {volume} {56}},\ \bibinfo {pages} {226} (\bibinfo {year} {2009})}\BibitemShut {NoStop}%
\bibitem [{\citenamefont {Conway}\ and\ \citenamefont {Kochen}(2011)}]{conway_kochen_2011}%
  \BibitemOpen
  \bibfield  {author} {\bibinfo {author} {\bibfnamefont {J.~H.}\ \bibnamefont {Conway}}\ and\ \bibinfo {author} {\bibfnamefont {S.}~\bibnamefont {Kochen}},\ }\bibinfo {title} {The strong free will theorem},\ in\ \href {https://doi.org/10.1017/CBO9780511976971.014} {\emph {\bibinfo {booktitle} {Deep Beauty: Understanding the Quantum World through Mathematical Innovation}}},\ \bibinfo {editor} {edited by\ \bibinfo {editor} {\bibfnamefont {H.}~\bibnamefont {Halvorson}}}\ (\bibinfo  {publisher} {Cambridge University Press},\ \bibinfo {year} {2011})\ pp.\ \bibinfo {pages} {443--454}\BibitemShut {NoStop}%
\bibitem [{\citenamefont {Peres}(1991)}]{Peres:1991JPA}%
  \BibitemOpen
  \bibfield  {author} {\bibinfo {author} {\bibfnamefont {A.}~\bibnamefont {Peres}},\ }\bibfield  {title} {\bibinfo {title} {Two simple proofs of the {K}ochen-{S}pecker theorem},\ }\href {https://doi.org/10.1088/0305-4470/24/4/003} {\bibfield  {journal} {\bibinfo  {journal} {J. Phys. A: Math. Gen.}\ }\textbf {\bibinfo {volume} {24}},\ \bibinfo {pages} {L175} (\bibinfo {year} {1991})}\BibitemShut {NoStop}%
\bibitem [{\citenamefont {Gould}\ and\ \citenamefont {Aravind}(2010)}]{gould2010isomorphism}%
  \BibitemOpen
  \bibfield  {author} {\bibinfo {author} {\bibfnamefont {E.}~\bibnamefont {Gould}}\ and\ \bibinfo {author} {\bibfnamefont {P.~K.}\ \bibnamefont {Aravind}},\ }\bibfield  {title} {\bibinfo {title} {Isomorphism between the {P}eres and {P}enrose proofs of the {BKS} theorem in three dimensions},\ }\href {https://doi.org/10.1007/s10701-010-9434-2} {\bibfield  {journal} {\bibinfo  {journal} {Found. Phys.}\ }\textbf {\bibinfo {volume} {40}},\ \bibinfo {pages} {1096} (\bibinfo {year} {2010})}\BibitemShut {NoStop}%
\bibitem [{\citenamefont {Bengtsson}(2012)}]{bengtsson2012gleason}%
  \BibitemOpen
  \bibfield  {author} {\bibinfo {author} {\bibfnamefont {I.}~\bibnamefont {Bengtsson}},\ }\bibfield  {title} {\bibinfo {title} {Gleason, {K}ochen-{S}pecker, and a competition that never was},\ }in\ \href {https://doi.org/10.1063/1.4773124} {\emph {\bibinfo {booktitle} {AIP Conference Proceedings}}},\ Vol.\ \bibinfo {volume} {1508}\ (\bibinfo {organization} {American Institute of Physics},\ \bibinfo {address} {Melville, NY},\ \bibinfo {year} {2012})\ pp.\ \bibinfo {pages} {125--135}\BibitemShut {NoStop}%
\bibitem [{\citenamefont {Penrose}(2000)}]{Penrose:2000}%
  \BibitemOpen
  \bibfield  {author} {\bibinfo {author} {\bibfnamefont {R.}~\bibnamefont {Penrose}},\ }\bibfield  {title} {\bibinfo {title} {On {B}ell non-locality without probabilities: Some curious geometry},\ }in\ \href@noop {} {\emph {\bibinfo {booktitle} {Quantum Reflections}}},\ \bibinfo {editor} {edited by\ \bibinfo {editor} {\bibfnamefont {J.}~\bibnamefont {Ellis}}\ and\ \bibinfo {editor} {\bibfnamefont {D.}~\bibnamefont {Amati}}}\ (\bibinfo  {publisher} {Cambridge University Press},\ \bibinfo {address} {Cambridge, UK},\ \bibinfo {year} {2000})\ pp.\ \bibinfo {pages} {1--27}\BibitemShut {NoStop}%
\bibitem [{\citenamefont {Peres}(1993)}]{Peres:1993}%
  \BibitemOpen
  \bibfield  {author} {\bibinfo {author} {\bibfnamefont {A.}~\bibnamefont {Peres}},\ }\href {https://doi.org/https://doi.org/10.1007/0-306-47120-5} {\emph {\bibinfo {title} {Quantum Theory: Concepts and Methods}}}\ (\bibinfo  {publisher} {Kluwer},\ \bibinfo {address} {Dordrecht},\ \bibinfo {year} {1993})\BibitemShut {NoStop}%
\bibitem [{\citenamefont {Cabello}\ and\ \citenamefont {Garc{\'\i}a-Alcaine}(1995)}]{Cabello:1995JPA}%
  \BibitemOpen
  \bibfield  {author} {\bibinfo {author} {\bibfnamefont {A.}~\bibnamefont {Cabello}}\ and\ \bibinfo {author} {\bibfnamefont {G.}~\bibnamefont {Garc{\'\i}a-Alcaine}},\ }\bibfield  {title} {\bibinfo {title} {A hidden-variables versus quantum mechanics experiment},\ }\href {https://doi.org/10.1088/0305-4470/28/13/016} {\bibfield  {journal} {\bibinfo  {journal} {J. Phys. A: Math. Gen.}\ }\textbf {\bibinfo {volume} {28}},\ \bibinfo {pages} {3719} (\bibinfo {year} {1995})}\BibitemShut {NoStop}%
\bibitem [{\citenamefont {Zimba}\ and\ \citenamefont {Penrose}(1993)}]{Zimba:1993SHPS}%
  \BibitemOpen
  \bibfield  {author} {\bibinfo {author} {\bibfnamefont {J.~R.}\ \bibnamefont {Zimba}}\ and\ \bibinfo {author} {\bibfnamefont {R.}~\bibnamefont {Penrose}},\ }\bibfield  {title} {\bibinfo {title} {On {B}ell non-locality without probabilities: More curious geometry},\ }\href {https://doi.org/10.1016/0039-3681(93)90061-N} {\bibfield  {journal} {\bibinfo  {journal} {Stud. Hist. Philos. Sci. A}\ }\textbf {\bibinfo {volume} {24}},\ \bibinfo {pages} {697} (\bibinfo {year} {1993})}\BibitemShut {NoStop}%
\bibitem [{\citenamefont {Cabello}\ and\ \citenamefont {García-Alcaine}(1996)}]{Cabello:1996JPA}%
  \BibitemOpen
  \bibfield  {author} {\bibinfo {author} {\bibfnamefont {A.}~\bibnamefont {Cabello}}\ and\ \bibinfo {author} {\bibfnamefont {G.}~\bibnamefont {García-Alcaine}},\ }\bibfield  {title} {\bibinfo {title} {{B}ell-{K}ochen-{S}pecker theorem for any finite dimensions {$n \geq 3$}},\ }\href {https://doi.org/10.1088/0305-4470/29/5/016} {\bibfield  {journal} {\bibinfo  {journal} {J. Phys. A: Math. Gen.}\ }\textbf {\bibinfo {volume} {29}},\ \bibinfo {pages} {1025} (\bibinfo {year} {1996})}\BibitemShut {NoStop}%
\bibitem [{\citenamefont {Cabello}\ \emph {et~al.}(2005)\citenamefont {Cabello}, \citenamefont {Estebaranz},\ and\ \citenamefont {Garc{\'\i}a-Alcaine}}]{Cabello:2005PLA}%
  \BibitemOpen
  \bibfield  {author} {\bibinfo {author} {\bibfnamefont {A.}~\bibnamefont {Cabello}}, \bibinfo {author} {\bibfnamefont {J.~M.}\ \bibnamefont {Estebaranz}},\ and\ \bibinfo {author} {\bibfnamefont {G.}~\bibnamefont {Garc{\'\i}a-Alcaine}},\ }\bibfield  {title} {\bibinfo {title} {Recursive proof of the {B}ell--{K}ochen--{S}pecker theorem in any dimension n\textgreater{}3},\ }\href {https://doi.org/10.1016/j.physleta.2005.03.067} {\bibfield  {journal} {\bibinfo  {journal} {Phys. Lett. A}\ }\textbf {\bibinfo {volume} {339}},\ \bibinfo {pages} {425} (\bibinfo {year} {2005})}\BibitemShut {NoStop}%
\bibitem [{\citenamefont {Matsuno}(2007)}]{MatsunoJPA2007}%
  \BibitemOpen
  \bibfield  {author} {\bibinfo {author} {\bibfnamefont {S.}~\bibnamefont {Matsuno}},\ }\bibfield  {title} {\bibinfo {title} {The construction of {K}ochen--{S}pecker noncolourable sets in higher-dimensional space from corresponding sets in lower dimension: modification of {C}abello, {E}stebaranz and {G}arc{\'\i}a-{A}lcaine's method},\ }\href {https://doi.org/10.1088/1751-8113/40/31/024} {\bibfield  {journal} {\bibinfo  {journal} {J. Phys. A: Math. Gen.}\ }\textbf {\bibinfo {volume} {40}},\ \bibinfo {pages} {9507} (\bibinfo {year} {2007})}\BibitemShut {NoStop}%
\bibitem [{\citenamefont {Cabello}\ \emph {et~al.}(2018)\citenamefont {Cabello}, \citenamefont {Portillo}, \citenamefont {Sol{\'\i}s},\ and\ \citenamefont {Svozil}}]{Cabello:2018PRA}%
  \BibitemOpen
  \bibfield  {author} {\bibinfo {author} {\bibfnamefont {A.}~\bibnamefont {Cabello}}, \bibinfo {author} {\bibfnamefont {J.~R.}\ \bibnamefont {Portillo}}, \bibinfo {author} {\bibfnamefont {A.}~\bibnamefont {Sol{\'\i}s}},\ and\ \bibinfo {author} {\bibfnamefont {K.}~\bibnamefont {Svozil}},\ }\bibfield  {title} {\bibinfo {title} {Minimal true-implies-false and true-implies-true sets of propositions in noncontextual hidden-variable theories},\ }\href {https://doi.org/10.1103/PhysRevA.98.012106} {\bibfield  {journal} {\bibinfo  {journal} {Phys. Rev. A}\ }\textbf {\bibinfo {volume} {98}},\ \bibinfo {pages} {012106} (\bibinfo {year} {2018})}\BibitemShut {NoStop}%
\bibitem [{\citenamefont {Ramanathan}\ \emph {et~al.}(2020)\citenamefont {Ramanathan}, \citenamefont {Rosicka}, \citenamefont {Horodecki}, \citenamefont {Pironio}, \citenamefont {Horodecki},\ and\ \citenamefont {Horodecki}}]{Ramanathan:2020Q}%
  \BibitemOpen
  \bibfield  {author} {\bibinfo {author} {\bibfnamefont {R.}~\bibnamefont {Ramanathan}}, \bibinfo {author} {\bibfnamefont {M.}~\bibnamefont {Rosicka}}, \bibinfo {author} {\bibfnamefont {K.}~\bibnamefont {Horodecki}}, \bibinfo {author} {\bibfnamefont {S.}~\bibnamefont {Pironio}}, \bibinfo {author} {\bibfnamefont {M.}~\bibnamefont {Horodecki}},\ and\ \bibinfo {author} {\bibfnamefont {P.}~\bibnamefont {Horodecki}},\ }\bibfield  {title} {\bibinfo {title} {Gadget structures in proofs of the {K}ochen-{S}pecker theorem},\ }\href {https://doi.org/10.22331/q-2020-08-14-308} {\bibfield  {journal} {\bibinfo  {journal} {{Quantum}}\ }\textbf {\bibinfo {volume} {4}},\ \bibinfo {pages} {308} (\bibinfo {year} {2020})}\BibitemShut {NoStop}%
\bibitem [{\citenamefont {Zhu}(2015)}]{Zhu:2015AP}%
  \BibitemOpen
  \bibfield  {author} {\bibinfo {author} {\bibfnamefont {H.}~\bibnamefont {Zhu}},\ }\bibfield  {title} {\bibinfo {title} {Super-symmetric informationally complete measurements},\ }\href {https://doi.org/https://doi.org/10.1016/j.aop.2015.08.005} {\bibfield  {journal} {\bibinfo  {journal} {Ann. Phys. (N. Y.)}\ }\textbf {\bibinfo {volume} {362}},\ \bibinfo {pages} {311} (\bibinfo {year} {2015})}\BibitemShut {NoStop}%
\bibitem [{\citenamefont {Gleason}\ and\ \citenamefont {Jost}(1976)}]{Jost:1976}%
  \BibitemOpen
  \bibfield  {author} {\bibinfo {author} {\bibfnamefont {A.~M.}\ \bibnamefont {Gleason}}\ and\ \bibinfo {author} {\bibfnamefont {R.}~\bibnamefont {Jost}},\ }\bibinfo {title} {Measures on the finite dimensional subspaces of a {H}ilbert space: Remarks to a theorem},\ in\ \href {https://doi.org/doi:10.1515/9781400868940-011} {\emph {\bibinfo {booktitle} {Stud. Math. Phys.}}},\ \bibinfo {editor} {edited by\ \bibinfo {editor} {\bibfnamefont {E.~H.}\ \bibnamefont {Lieb}}}\ (\bibinfo  {publisher} {Princeton University Press},\ \bibinfo {address} {Princeton},\ \bibinfo {year} {1976})\ pp.\ \bibinfo {pages} {209--228}\BibitemShut {NoStop}%
\bibitem [{\citenamefont {Peres}\ and\ \citenamefont {Ron}(1988)}]{Peres:1988}%
  \BibitemOpen
  \bibfield  {author} {\bibinfo {author} {\bibfnamefont {A.}~\bibnamefont {Peres}}\ and\ \bibinfo {author} {\bibfnamefont {A.}~\bibnamefont {Ron}},\ }\bibfield  {title} {\bibinfo {title} {Cryptodeterminism and quantum theory},\ }in\ \href@noop {} {\emph {\bibinfo {booktitle} {Microphysical Reality and Quantum Formalism}}},\ Vol.~\bibinfo {volume} {2},\ \bibinfo {editor} {edited by\ \bibinfo {editor} {\bibfnamefont {A.}~\bibnamefont {van~der Merwe}}, \bibinfo {editor} {\bibfnamefont {F.}~\bibnamefont {Selleri}},\ and\ \bibinfo {editor} {\bibfnamefont {G.}~\bibnamefont {Tarozzi}}}\ (\bibinfo  {publisher} {Kluwer},\ \bibinfo {address} {Dordrecht},\ \bibinfo {year} {1988})\ pp.\ \bibinfo {pages} {115--123}\BibitemShut {NoStop}%
\bibitem [{\citenamefont {Pavi\v{c}i{\'c}}\ \emph {et~al.}(2005)\citenamefont {Pavi\v{c}i{\'c}}, \citenamefont {Merlet}, \citenamefont {McKay},\ and\ \citenamefont {Megill}}]{Pavicic:2005JPA}%
  \BibitemOpen
  \bibfield  {author} {\bibinfo {author} {\bibfnamefont {M.}~\bibnamefont {Pavi\v{c}i{\'c}}}, \bibinfo {author} {\bibfnamefont {J.-P.}\ \bibnamefont {Merlet}}, \bibinfo {author} {\bibfnamefont {B.~D.}\ \bibnamefont {McKay}},\ and\ \bibinfo {author} {\bibfnamefont {N.~D.}\ \bibnamefont {Megill}},\ }\bibfield  {title} {\bibinfo {title} {{K}ochen-{S}pecker vectors},\ }\href {https://doi.org/10.1088/0305-4470/38/7/013} {\bibfield  {journal} {\bibinfo  {journal} {J. Phys. A: Math. Gen.}\ }\textbf {\bibinfo {volume} {38}},\ \bibinfo {pages} {1577} (\bibinfo {year} {2005})}\BibitemShut {NoStop}%
\bibitem [{\citenamefont {Arends}\ \emph {et~al.}(2011)\citenamefont {Arends}, \citenamefont {Ouaknine},\ and\ \citenamefont {Wampler}}]{Arends:2011}%
  \BibitemOpen
  \bibfield  {author} {\bibinfo {author} {\bibfnamefont {F.}~\bibnamefont {Arends}}, \bibinfo {author} {\bibfnamefont {J.}~\bibnamefont {Ouaknine}},\ and\ \bibinfo {author} {\bibfnamefont {C.~W.}\ \bibnamefont {Wampler}},\ }\bibfield  {title} {\bibinfo {title} {On searching for small {K}ochen-{S}pecker vector systems},\ }in\ \href@noop {} {\emph {\bibinfo {booktitle} {Graph-Theoretic Concepts in Computer Science}}},\ \bibinfo {editor} {edited by\ \bibinfo {editor} {\bibfnamefont {P.}~\bibnamefont {Kolman}}\ and\ \bibinfo {editor} {\bibfnamefont {J.}~\bibnamefont {Kratochv{\'i}l}}}\ (\bibinfo  {publisher} {Springer Berlin Heidelberg},\ \bibinfo {address} {Berlin, Heidelberg},\ \bibinfo {year} {2011})\ pp.\ \bibinfo {pages} {23--34}\BibitemShut {NoStop}%
\bibitem [{\citenamefont {Uijlen}\ and\ \citenamefont {Westerbaan}(2016)}]{Uijlen:2016}%
  \BibitemOpen
  \bibfield  {author} {\bibinfo {author} {\bibfnamefont {S.}~\bibnamefont {Uijlen}}\ and\ \bibinfo {author} {\bibfnamefont {B.}~\bibnamefont {Westerbaan}},\ }\bibfield  {title} {\bibinfo {title} {A {K}ochen-{S}pecker system has at least 22 vectors},\ }\bibfield  {journal} {\bibinfo  {journal} {New Gener. Comput.}\ }\textbf {\bibinfo {volume} {34}},\ \href {https://doi.org/10.1007/s00354-016-0202-5} {10.1007/s00354-016-0202-5} (\bibinfo {year} {2016})\BibitemShut {NoStop}%
\bibitem [{\citenamefont {Williams}\ and\ \citenamefont {Constantin}(2025)}]{williams2024maximalnonkochenspeckersetslower}%
  \BibitemOpen
  \bibfield  {author} {\bibinfo {author} {\bibfnamefont {T.}~\bibnamefont {Williams}}\ and\ \bibinfo {author} {\bibfnamefont {A.}~\bibnamefont {Constantin}},\ }\href {https://doi.org/10.1103/PhysRevA.111.012223} {\bibinfo {title} {Maximal non-{K}ochen-{S}pecker sets and a lower bound on the size of {K}ochen-{S}pecker sets}} (\bibinfo {year} {2025})\BibitemShut {NoStop}%
\bibitem [{\citenamefont {Renes}\ \emph {et~al.}(2004)\citenamefont {Renes}, \citenamefont {Blume-Kohout}, \citenamefont {Scott},\ and\ \citenamefont {Caves}}]{Renes2004JMP}%
  \BibitemOpen
  \bibfield  {author} {\bibinfo {author} {\bibfnamefont {J.~M.}\ \bibnamefont {Renes}}, \bibinfo {author} {\bibfnamefont {R.}~\bibnamefont {Blume-Kohout}}, \bibinfo {author} {\bibfnamefont {A.~J.}\ \bibnamefont {Scott}},\ and\ \bibinfo {author} {\bibfnamefont {C.~M.}\ \bibnamefont {Caves}},\ }\bibfield  {title} {\bibinfo {title} {Symmetric informationally complete quantum measurements},\ }\href {https://doi.org/10.1063/1.1737053} {\bibfield  {journal} {\bibinfo  {journal} {J. Math. Phys.}\ }\textbf {\bibinfo {volume} {45}},\ \bibinfo {pages} {2171} (\bibinfo {year} {2004})}\BibitemShut {NoStop}%
\bibitem [{\citenamefont {Fuchs}\ \emph {et~al.}(2017)\citenamefont {Fuchs}, \citenamefont {Hoang},\ and\ \citenamefont {Stacey}}]{Fuchs:2017A}%
  \BibitemOpen
  \bibfield  {author} {\bibinfo {author} {\bibfnamefont {C.~A.}\ \bibnamefont {Fuchs}}, \bibinfo {author} {\bibfnamefont {M.~C.}\ \bibnamefont {Hoang}},\ and\ \bibinfo {author} {\bibfnamefont {B.~C.}\ \bibnamefont {Stacey}},\ }\bibfield  {title} {\bibinfo {title} {The {SIC} question: History and state of play},\ }\bibfield  {journal} {\bibinfo  {journal} {Axioms}\ }\textbf {\bibinfo {volume} {6}},\ \href {https://doi.org/10.3390/axioms6030021} {10.3390/axioms6030021} (\bibinfo {year} {2017})\BibitemShut {NoStop}%
\bibitem [{\citenamefont {DeBrota}\ \emph {et~al.}(2020)\citenamefont {DeBrota}, \citenamefont {Fuchs},\ and\ \citenamefont {Stacey}}]{DeBrota:2020PRR}%
  \BibitemOpen
  \bibfield  {author} {\bibinfo {author} {\bibfnamefont {J.~B.}\ \bibnamefont {DeBrota}}, \bibinfo {author} {\bibfnamefont {C.~A.}\ \bibnamefont {Fuchs}},\ and\ \bibinfo {author} {\bibfnamefont {B.~C.}\ \bibnamefont {Stacey}},\ }\bibfield  {title} {\bibinfo {title} {Symmetric informationally complete measurements identify the irreducible difference between classical and quantum systems},\ }\href {https://doi.org/10.1103/PhysRevResearch.2.013074} {\bibfield  {journal} {\bibinfo  {journal} {Phys. Rev. Res.}\ }\textbf {\bibinfo {volume} {2}},\ \bibinfo {pages} {013074} (\bibinfo {year} {2020})}\BibitemShut {NoStop}%
\bibitem [{\citenamefont {Szöllősi}(2014)}]{Szollosi:2014XXX}%
  \BibitemOpen
  \bibfield  {author} {\bibinfo {author} {\bibfnamefont {F.}~\bibnamefont {Szöllősi}},\ }\href {https://arxiv.org/abs/1402.6429} {\bibinfo {title} {All complex equiangular tight frames in dimension 3}} (\bibinfo {year} {2014}),\ \Eprint {https://arxiv.org/abs/1402.6429} {arXiv:1402.6429 [math.FA]} \BibitemShut {NoStop}%
\bibitem [{\citenamefont {Hughston}\ and\ \citenamefont {Salamon}(2016)}]{Hughston:2016AM}%
  \BibitemOpen
  \bibfield  {author} {\bibinfo {author} {\bibfnamefont {L.~P.}\ \bibnamefont {Hughston}}\ and\ \bibinfo {author} {\bibfnamefont {S.~M.}\ \bibnamefont {Salamon}},\ }\bibfield  {title} {\bibinfo {title} {Surveying points in the complex projective plane},\ }\href {https://doi.org/https://doi.org/10.1016/j.aim.2015.09.022} {\bibfield  {journal} {\bibinfo  {journal} {Adv. Math.}\ }\textbf {\bibinfo {volume} {286}},\ \bibinfo {pages} {1017} (\bibinfo {year} {2016})}\BibitemShut {NoStop}%
\bibitem [{\citenamefont {Wootters}(2006)}]{Wootters:2006quantum}%
  \BibitemOpen
  \bibfield  {author} {\bibinfo {author} {\bibfnamefont {W.~K.}\ \bibnamefont {Wootters}},\ }\bibfield  {title} {\bibinfo {title} {Quantum measurements and finite geometry},\ }\href {https://doi.org/https://doi.org/10.1007/s10701-005-9008-x} {\bibfield  {journal} {\bibinfo  {journal} {Found. Phys.}\ }\textbf {\bibinfo {volume} {36}},\ \bibinfo {pages} {112} (\bibinfo {year} {2006})}\BibitemShut {NoStop}%
\bibitem [{\citenamefont {Salt}(2023)}]{Salt:2023}%
  \BibitemOpen
  \bibfield  {author} {\bibinfo {author} {\bibfnamefont {L.}~\bibnamefont {Salt}},\ }\emph {\bibinfo {title} {New proofs of the Kochen--Specker theorem via Hadamard matrices}},\ \href {https://summit.sfu.ca/item/36526} {Master's thesis},\ \bibinfo  {school} {Simon Fraser University} (\bibinfo {year} {2023})\BibitemShut {NoStop}%
\bibitem [{\citenamefont {Bub}(1996)}]{Bub:1996FP}%
  \BibitemOpen
  \bibfield  {author} {\bibinfo {author} {\bibfnamefont {J.}~\bibnamefont {Bub}},\ }\bibfield  {title} {\bibinfo {title} {Sch\"{u}tte's tautology and the {K}ochen-{S}pecker theorem},\ }\href {https://doi.org/10.1007/bf02058633} {\bibfield  {journal} {\bibinfo  {journal} {Found. Phys.}\ }\textbf {\bibinfo {volume} {26}},\ \bibinfo {pages} {787} (\bibinfo {year} {1996})}\BibitemShut {NoStop}%
\bibitem [{\citenamefont {Cabello}(2017)}]{Cabello:2017}%
  \BibitemOpen
  \bibfield  {author} {\bibinfo {author} {\bibfnamefont {A.}~\bibnamefont {Cabello}},\ }\bibinfo {title} {The {U}nspeakable {W}hy},\ in\ \href {https://doi.org/10.1007/978-3-319-38987-5_11} {\emph {\bibinfo {booktitle} {Quantum [Un]Speakables II: Half a Century of {B}ell's Theorem}}},\ \bibinfo {editor} {edited by\ \bibinfo {editor} {\bibfnamefont {R.}~\bibnamefont {Bertlmann}}\ and\ \bibinfo {editor} {\bibfnamefont {A.}~\bibnamefont {Zeilinger}}}\ (\bibinfo  {publisher} {Springer International Publishing},\ \bibinfo {address} {Cham},\ \bibinfo {year} {2017})\ pp.\ \bibinfo {pages} {189--199}\BibitemShut {NoStop}%
\bibitem [{\citenamefont {Xu}\ \emph {et~al.}(2015)\citenamefont {Xu}, \citenamefont {Chen},\ and\ \citenamefont {Su}}]{Xu:2015PLA}%
  \BibitemOpen
  \bibfield  {author} {\bibinfo {author} {\bibfnamefont {Z.-P.}\ \bibnamefont {Xu}}, \bibinfo {author} {\bibfnamefont {J.-L.}\ \bibnamefont {Chen}},\ and\ \bibinfo {author} {\bibfnamefont {H.-Y.}\ \bibnamefont {Su}},\ }\bibfield  {title} {\bibinfo {title} {State-independent contextuality sets for a qutrit},\ }\href {https://doi.org/10.1016/j.physleta.2015.04.024} {\bibfield  {journal} {\bibinfo  {journal} {Phys. Lett. A}\ }\textbf {\bibinfo {volume} {379}},\ \bibinfo {pages} {1868} (\bibinfo {year} {2015})}\BibitemShut {NoStop}%
\bibitem [{\citenamefont {Cabello}\ \emph {et~al.}(2015)\citenamefont {Cabello}, \citenamefont {Kleinmann},\ and\ \citenamefont {Budroni}}]{CKB2015}%
  \BibitemOpen
  \bibfield  {author} {\bibinfo {author} {\bibfnamefont {A.}~\bibnamefont {Cabello}}, \bibinfo {author} {\bibfnamefont {M.}~\bibnamefont {Kleinmann}},\ and\ \bibinfo {author} {\bibfnamefont {C.}~\bibnamefont {Budroni}},\ }\bibfield  {title} {\bibinfo {title} {Necessary and sufficient condition for quantum state-independent contextuality},\ }\href {https://doi.org/10.1103/PhysRevLett.114.250402} {\bibfield  {journal} {\bibinfo  {journal} {Phys. Rev. Lett.}\ }\textbf {\bibinfo {volume} {114}},\ \bibinfo {pages} {250402} (\bibinfo {year} {2015})}\BibitemShut {NoStop}%
\bibitem [{\citenamefont {Budroni}\ \emph {et~al.}(2022)\citenamefont {Budroni}, \citenamefont {Cabello}, \citenamefont {G{\"u}hne}, \citenamefont {Kleinmann},\ and\ \citenamefont {Larsson}}]{Budroni:2022RMP}%
  \BibitemOpen
  \bibfield  {author} {\bibinfo {author} {\bibfnamefont {C.}~\bibnamefont {Budroni}}, \bibinfo {author} {\bibfnamefont {A.}~\bibnamefont {Cabello}}, \bibinfo {author} {\bibfnamefont {O.}~\bibnamefont {G{\"u}hne}}, \bibinfo {author} {\bibfnamefont {M.}~\bibnamefont {Kleinmann}},\ and\ \bibinfo {author} {\bibfnamefont {J.-{\AA}.}\ \bibnamefont {Larsson}},\ }\bibfield  {title} {\bibinfo {title} {{K}ochen-{S}pecker contextuality},\ }\href {https://doi.org/10.1103/RevModPhys.94.045007} {\bibfield  {journal} {\bibinfo  {journal} {Rev. Mod. Phys.}\ }\textbf {\bibinfo {volume} {94}},\ \bibinfo {pages} {045007} (\bibinfo {year} {2022})}\BibitemShut {NoStop}%
\bibitem [{\citenamefont {Cabello}(1996)}]{Cabello:1996}%
  \BibitemOpen
  \bibfield  {author} {\bibinfo {author} {\bibfnamefont {A.}~\bibnamefont {Cabello}},\ }\emph {\bibinfo {title} {Pruebas algebraicas de imposibilidad de variables ocultas en mec\'anica cu\'antica}},\ \href {https://hdl.handle.net/20.500.14352/62875} {Ph.D. thesis},\ \bibinfo  {school} {Universidad Complutense de Madrid} (\bibinfo {year} {1996})\BibitemShut {NoStop}%
\bibitem [{\citenamefont {Chalupa}\ \emph {et~al.}(1979)\citenamefont {Chalupa}, \citenamefont {Leath},\ and\ \citenamefont {Reich}}]{Chalupa:1979}%
  \BibitemOpen
  \bibfield  {author} {\bibinfo {author} {\bibfnamefont {J.}~\bibnamefont {Chalupa}}, \bibinfo {author} {\bibfnamefont {P.~L.}\ \bibnamefont {Leath}},\ and\ \bibinfo {author} {\bibfnamefont {G.~R.}\ \bibnamefont {Reich}},\ }\bibfield  {title} {\bibinfo {title} {Bootstrap percolation on a {B}ethe lattice},\ }\href {https://doi.org/10.1088/0022-3719/12/1/008} {\bibfield  {journal} {\bibinfo  {journal} {J. Phys. C}\ }\textbf {\bibinfo {volume} {12}},\ \bibinfo {pages} {L31} (\bibinfo {year} {1979})}\BibitemShut {NoStop}%
\bibitem [{\citenamefont {Balogh}\ and\ \citenamefont {Pete}(1998)}]{Balogh:1998}%
  \BibitemOpen
  \bibfield  {author} {\bibinfo {author} {\bibfnamefont {J.}~\bibnamefont {Balogh}}\ and\ \bibinfo {author} {\bibfnamefont {G.}~\bibnamefont {Pete}},\ }\bibfield  {title} {\bibinfo {title} {Random disease on the square grid},\ }\href {https://doi.org/https://doi.org/10.1002/(SICI)1098-2418(199810/12)13:3/4<409::AID-RSA11>3.0.CO;2-U} {\bibfield  {journal} {\bibinfo  {journal} {Random Struct. Algor.}\ }\textbf {\bibinfo {volume} {13}},\ \bibinfo {pages} {409} (\bibinfo {year} {1998})}\BibitemShut {NoStop}%
\bibitem [{\citenamefont {Baxter}\ \emph {et~al.}(2010)\citenamefont {Baxter}, \citenamefont {Dorogovtsev}, \citenamefont {Goltsev},\ and\ \citenamefont {Mendes}}]{Baxter:2010}%
  \BibitemOpen
  \bibfield  {author} {\bibinfo {author} {\bibfnamefont {G.~J.}\ \bibnamefont {Baxter}}, \bibinfo {author} {\bibfnamefont {S.~N.}\ \bibnamefont {Dorogovtsev}}, \bibinfo {author} {\bibfnamefont {A.~V.}\ \bibnamefont {Goltsev}},\ and\ \bibinfo {author} {\bibfnamefont {J.~F.~F.}\ \bibnamefont {Mendes}},\ }\bibfield  {title} {\bibinfo {title} {Bootstrap percolation on complex networks},\ }\href {https://doi.org/10.1103/PhysRevE.82.011103} {\bibfield  {journal} {\bibinfo  {journal} {Phys. Rev. E}\ }\textbf {\bibinfo {volume} {82}},\ \bibinfo {pages} {011103} (\bibinfo {year} {2010})}\BibitemShut {NoStop}%
\bibitem [{\citenamefont {Balogh}\ and\ \citenamefont {Pittel}(2007)}]{Balogh:2007}%
  \BibitemOpen
  \bibfield  {author} {\bibinfo {author} {\bibfnamefont {J.}~\bibnamefont {Balogh}}\ and\ \bibinfo {author} {\bibfnamefont {B.~G.}\ \bibnamefont {Pittel}},\ }\bibfield  {title} {\bibinfo {title} {Bootstrap percolation on the random regular graph},\ }\href {https://doi.org/https://doi.org/10.1002/rsa.20158} {\bibfield  {journal} {\bibinfo  {journal} {Random Struct. Algor.}\ }\textbf {\bibinfo {volume} {30}},\ \bibinfo {pages} {257} (\bibinfo {year} {2007})}\BibitemShut {NoStop}%
\bibitem [{\citenamefont {Balogh}\ \emph {et~al.}(2012)\citenamefont {Balogh}, \citenamefont {Bollob{\'a}s}, \citenamefont {Duminil-Copin},\ and\ \citenamefont {Morris}}]{Balogh:2012}%
  \BibitemOpen
  \bibfield  {author} {\bibinfo {author} {\bibfnamefont {J.}~\bibnamefont {Balogh}}, \bibinfo {author} {\bibfnamefont {B.}~\bibnamefont {Bollob{\'a}s}}, \bibinfo {author} {\bibfnamefont {H.}~\bibnamefont {Duminil-Copin}},\ and\ \bibinfo {author} {\bibfnamefont {R.}~\bibnamefont {Morris}},\ }\bibfield  {title} {\bibinfo {title} {The sharp threshold for bootstrap percolation in all dimensions},\ }\href {https://doi.org/https://doi.org/10.1090/S0002-9947-2011-05552-2} {\bibfield  {journal} {\bibinfo  {journal} {Trans. Am. Math. Soc.}\ }\textbf {\bibinfo {volume} {364}},\ \bibinfo {pages} {2667} (\bibinfo {year} {2012})}\BibitemShut {NoStop}%
\bibitem [{\citenamefont {Morris}(2017)}]{Morris:2017}%
  \BibitemOpen
  \bibfield  {author} {\bibinfo {author} {\bibfnamefont {R.}~\bibnamefont {Morris}},\ }\bibfield  {title} {\bibinfo {title} {Bootstrap percolation, and other automata},\ }\href {https://doi.org/https://doi.org/10.1016/j.ejc.2017.06.024} {\bibfield  {journal} {\bibinfo  {journal} {Eur. J. of Comb.}\ }\textbf {\bibinfo {volume} {66}},\ \bibinfo {pages} {250} (\bibinfo {year} {2017})}\BibitemShut {NoStop}%
\bibitem [{\citenamefont {Reichman}(2012)}]{Reichman:2012}%
  \BibitemOpen
  \bibfield  {author} {\bibinfo {author} {\bibfnamefont {D.}~\bibnamefont {Reichman}},\ }\bibfield  {title} {\bibinfo {title} {New bounds for contagious sets},\ }\href {https://doi.org/https://doi.org/10.1016/j.disc.2012.01.016} {\bibfield  {journal} {\bibinfo  {journal} {Discrete Math.}\ }\textbf {\bibinfo {volume} {312}},\ \bibinfo {pages} {1812} (\bibinfo {year} {2012})}\BibitemShut {NoStop}%
\bibitem [{\citenamefont {Wesolek}(2019)}]{Wesolek:2019}%
  \BibitemOpen
  \bibfield  {author} {\bibinfo {author} {\bibfnamefont {A.}~\bibnamefont {Wesolek}},\ }\href {https://arxiv.org/abs/1909.04649} {\bibinfo {title} {Bootstrap percolation in ore-type graphs}} (\bibinfo {year} {2019}),\ \Eprint {https://arxiv.org/abs/1909.04649} {1909.04649} \BibitemShut {NoStop}%
\bibitem [{\citenamefont {Gunderson}(2020)}]{Gunderson:2020}%
  \BibitemOpen
  \bibfield  {author} {\bibinfo {author} {\bibfnamefont {K.}~\bibnamefont {Gunderson}},\ }\bibfield  {title} {\bibinfo {title} {Minimum degree conditions for small percolating sets in bootstrap percolation},\ }\href {https://doi.org/10.37236/6937} {\bibfield  {journal} {\bibinfo  {journal} {Electron. J. Comb.}\ }\textbf {\bibinfo {volume} {27}},\ \bibinfo {pages} {Paper No. 2.37, 22} (\bibinfo {year} {2020})}\BibitemShut {NoStop}%
\end{thebibliography}

%


\end{document}